\documentclass[11pt]{article}

\usepackage{vmargin,graphicx,amsmath,amssymb,color,subfigure,enumerate,csquotes, hyperref}
\newtheorem{theorem}{Theorem}[section]
\newtheorem{lemma}[theorem]{Lemma}

\newtheorem{proposition}[theorem]{Proposition}
\newtheorem{corollary}[theorem]{Corollary}

\newtheorem{assumption}[theorem]{Assumption}

\def\blacksquare{
\thinspace\nobreak \vrule width 5pt height 5pt depth 0pt}

\newenvironment{proof}{\begin{trivlist}
                       \item[]\hspace{0cm}{\bf Proof. }
                       \hspace{0cm} }{\hfill $\blacksquare$
                     \end{trivlist}}
\newenvironment{proofof}[1]{\begin{trivlist}
                       \item[]\hspace{0cm}{\bf Proof of #1. }
                       \hspace{0cm} }{\hfill $\blacksquare$
                     \end{trivlist}}

\makeatletter

\@addtoreset{equation}{section}
\makeatother

\definecolor{gr}{rgb}   {0.,   0.69,   0.23 }
\definecolor{bl}{rgb}   {0.,   0.5,   1. }
\definecolor{mg}{rgb}   {0.85,  0.,    0.85}
\definecolor{yl}{rgb}   {0.8,  0.7,   0.}

\newcommand{\dx}{\mathrm{d}\,}
\newcommand{\sL}{\mathsf{L}}
\newcommand{\sH}{\mathsf{H}}
\newcommand{\sW}{\mathsf{W}}

\newcommand{\phy}{\varphi}

\def\Dir{\mathsf{Dir}}
\def\dx{\,\mathsf{d}}
\def\R{\,\mathbb{R}}
\def\M{\,\mathbb{M}}
\def\Z{\,\mathbb{Z}}
\def\T{\,\mathbb{T}}
\def\eps{\varepsilon}
\def\dr{\partial}

\title{Strong confinement limit for the nonlinear Schr\"odinger equation constrained on a curve}
\author{F. M\'ehats and N. Raymond}

\begin{document}
\maketitle

\begin{abstract}
This paper is devoted to the cubic nonlinear Schr\"odinger equation in a two dimensional waveguide with shrinking cross section of order $\eps$. For a Cauchy data living essentially on the first mode of the transverse Laplacian, we provide a tensorial approximation of the solution $\psi^\eps$ in the limit $\eps\to 0$, with an estimate of the approximation error, and derive a limiting nonlinear Schr\"odinger equation in dimension one. If the Cauchy data $\psi^\eps_0$ has a uniformly bounded energy, then it is a bounded sequence in $\sH^1$ and we show that the approximation is of order $\mathcal O(\sqrt{\eps})$. If we assume that $\psi^\eps_0$ is bounded in the graph norm of the Hamiltonian, then it is a bounded sequence in $\sH^2$ and we show that the approximation error is of order $\mathcal O(\eps)$.
\end{abstract}

\sloppy

\section{Motivation and results}

\subsection{Motivation}

The Dirichlet realization of the Laplacian on tubes of the Euclidean space plays an important role in the physical description of nanostructures. In the last twenty years, many papers were concerned by the influence of the geometry of the tube (curvature, torsion) on the spectrum. For instance, in \cite{Duclos95}, Duclos and Exner proved that bending a waveguide in dimension two and three always induces the existence of discrete spectrum below the essential spectrum (see also \cite{Duclos05}). Another question of interest in their paper is the limit when the cross section shrinks to a point. In particular they prove that, in some sense, the Dirichlet Laplacian on a bidimensional tube, with cross section $(-\eps,\eps)$ is well approximated by Schr¬\"odinger operator
$$-\partial_{x_{1}}^2-\frac{\kappa^2(x_{1})}{4}-\frac{1}{\eps^2}\partial_{x_{2}}^2,$$
acting on $\sL^2(\R\times(-1,1),\dx x_{1} \dx x_{2})$ and where $\kappa$ denotes the curvature of the center line of the tube. Such approximations have been recently considered in \cite{SK12} or in presence of magnetic fields \cite{KR14} through a convergence of resolvent method. Concerning this kind of results, one may refer to the memoir by Wachsmuth and Teufel \cite{WT14} where dynamical problems are analyzed in the spirit of adiabatic reductions.

In the present paper, we will consider the time dependent Schr\"odinger equation with a cubic non linearity in a waveguide and we would especially like to determine if the adiabatic reduction available in the linear framework can be used to reduce the dimension of the non linear equation and provide an effective dynamics in dimension one. The derivation of nonlinear quantum models in reduced dimensions has been the object of several works in the past years. For the modeling of the dynamics of electrons in nanostructures, the dimension reduction problem for the Schr\"odinger-Poisson system has been studied in \cite{SP1,SP4} for confinement on the plane, in \cite{SP2} for confinement on a line, and in \cite{SP3} for confinement on the sphere. For the modeling of strongly anisotropic Bose-Einstein condensates, the case of cubic nonlinear Schr\"odinger equations with an harmonic potential has been considered in \cite{BEC1,BEC2,BEC3,BEC4,BEC5}.

\subsection{Geometry and normal form}
Let us describe the geometrical context of this paper. With the same formalism, we will consider the case of unbounded curves and the case of closed curves. Consider a smooth, simple curve $\Gamma$ in $\R^2$ defined by its normal parametrization $\gamma : x_{1}\mapsto \gamma(x_{1})$. For $\eps>0$ we introduce the map
\begin{equation}\label{Phi}
\Phi_{\eps} : \mathcal{S}=\M\times (-1,1)\ni(x_{1},x_{2})\mapsto \gamma(x_{1})+\eps x_{2}\nu(x_{1})=\mathsf{x},
\end{equation}
where $\nu(x_{1})$ denotes the unit normal vector at the point $\gamma(x_{1})$ such that $\det(\gamma'(x_{1}),\nu(x_{1}))=1$ and where
$$
\M=\left\{\begin{array}{ll}\R\quad&\mbox{for an unbounded curve,}\\
\T=\R/(2\pi\Z)\quad&\mbox{for a closed curve.}
\end{array}\right.
$$
We recall that the curvature at the point $\gamma(x_{1})$, denoted by $\kappa(x_{1})$, is defined by
$$\gamma''(x_{1})=\kappa(x_{1})\nu(x_{1}).$$
The waveguide is $\Omega_{\eps}=\Phi_{\eps}(\mathcal{S})$ and we will work under the following assumption which states that waveguide does not overlap itself and that $\Phi_{\eps}$ is a smooth diffeomorphism.
\begin{assumption}
\label{assumption1}
We assume that the function $\kappa$ is bounded, as well as its derivatives $\kappa'$ and $\kappa''$. Moreover, we assume that there exists $\eps_{0}\in (0,\frac{1}{\|\kappa\|_{\sL^\infty}})$ such that, for $\eps\in(0,\eps_{0})$, $\Phi_{\eps}$ is injective. 
\end{assumption}
We will denote by $-\Delta_{\Omega_{\eps}}^\Dir$ the Dirichlet Laplacian on $\Omega_{\eps}$. We are interested in the following equation:
\begin{equation}\label{CNLS}
i\dr_{t}\psi^\eps=-\Delta_{\Omega_{\eps}}^\Dir\psi^\eps+\lambda \eps^\alpha |\psi^\eps|^2\psi^\eps
\end{equation}
on $\Omega_\eps$ with a Cauchy condition $\psi^\eps(0;\cdot)=\psi^\eps_{0}$ and where $\alpha\geq 1$ and $\lambda\in \R$ are parameters.

By using the diffeomorphism $\Phi_{\eps}$, we may rewrite \eqref{CNLS} in the space coordinates $(x_{1},x_{2})$ given by \eqref{Phi}. For that purpose, let us introduce $m_{\eps}(x_{1},x_{2})=1-\eps x_{2}\kappa(x_{1})$ and consider the function $\psi^\eps$ transported by $\Phi_{\eps}$, 
$$\mathcal{U}_{\eps}\psi^\eps(t; x_{1},x_{2})=\phi^\eps(t ; x_{1}, x_{2})=\eps^{1/2}m_{\eps}(x_{1},x_{2})^{1/2}\psi^\eps(t ; \Phi_{\eps}(x_{1},x_{2})).$$
Note that $\mathcal{U}_{\eps}$ is unitary from $\sL^2(\Omega_{\eps}, \dx \mathsf{x})$ to $\sL^2(\mathcal{S}, \dx x_{1} \dx x_{2})$ and maps $\sH^1_0(\Omega_{\eps})$ (resp. $\sH^2(\Omega_\eps)$) to $\sH^1_0(\mathcal{S})$ (resp. to $\sH^2(\mathcal{S}))$. Moreover, the operator $-\Delta_{\Omega_{\eps}}^\Dir$ is unitarily equivalent to the self-adjoint operator on $\sL^2(\mathcal{S}, \dx x_{1} \dx x_{2})$,
$$\mathcal{U}_{\eps}(-\Delta_{\Omega_{\eps}^\Dir})\mathcal{U}^{-1}_{\eps}=\mathcal{H}_{\eps}+V_{\eps},\quad\mbox{ with } \mathcal{H}_{\eps}=\mathcal{P}^2_{\eps,1}+\mathcal{P}^2_{\eps,2}\,,$$
where $$\mathcal{P}_{\eps,1}=m_{\eps}^{-1/2}D_{x_{1}}m_{\eps}^{-1/2}, \qquad \mathcal{P}_{\eps,2}=\eps^{-1}D_{x_{2}}$$ and where the effective electric $V_\eps$ potential is defined by
$$V_{\eps}(x_{1},x_{2})=-\frac{\kappa(x_{1})^2}{4(1-\eps x_{2}\kappa(x_{1}))^{2}}.$$
We have used the standard notation $D=-i\dr$. Notice that, for all $\eps<\eps_0$, we have $m_\eps\geq 1-\eps_0\|\kappa\|_{\sL^\infty}>0$. The problem \eqref{CNLS} becomes
\begin{equation}\label{CNLS'}
i\dr_{t}\phi^\eps=\mathcal{H}_{\eps}\phi^\eps+V_{\eps}\phi^\eps+\lambda \eps^{\alpha-1}m_{\eps}^{-1} |\phi^\eps|^2\phi^\eps
\end{equation}
with Dirichlet boundary conditions $\phi^\eps(t;x_{1},\pm 1)=0$ and the Cauchy condition
$ \phi^\eps(\cdot ; 0)=\phi_{0}^\eps=\mathcal{U}_{\eps}\psi^\eps_0$. We notice that the domains of $\mathcal{H}_{\eps}$ and  $\mathcal{H}_{\eps}+V_{\eps}$ coincide with $\sH^2(\mathcal{S})\cap \sH^1_{0}(\mathcal{S})$.

\bigskip
In this paper we will analyze the critical case $\alpha=1$ where the nonlinear term is of the same order as the parallel kinetic energy associated to $D_{x_1}^2$. It is well-known that \eqref{CNLS} (thus \eqref{CNLS'} also) has two conserved quantities: the $\sL^2$ norm and the nonlinear energy. Let us introduce the first eigenvalue $\mu_1=\frac{\pi^2}{4}$ of $D_{x_2}^2$ on $(-1,1)$ with Dirichlet boundary conditions, associated to the eigenfunction $e_1(x_{2})=\cos\left(\frac{\pi}{2}x_{2}\right)$ and define the energy functional
\begin{align}\label{Energy}
\mathcal{E}_{\eps}(\phi)
&=\frac{1}{2}\int_{\mathcal{S}}|\mathcal{P}_{\eps,1}\phi|^2 \dx x_{1} \dx x_{2}+\frac{1}{2}\int_{\mathcal{S}}|\mathcal{P}_{\eps,2}\phi|^2 \dx x_{1} \dx x_{2}+\frac{1}{2}\int_{\mathcal{S}}\left(V_{\eps}-\frac{\mu_1}{\eps^2}\right)|\phi|^2 \dx x_{1} \dx x_{2}\nonumber\\
&\qquad +\frac{\lambda}{4}\int_{\mathcal{S}}m_{\eps}^{-1}|\phi|^4 \dx x_{1} \dx x_{2}\\&=\frac{1}{2}\int_{\Omega_\eps}|\nabla (\mathcal{U}^{-1}_{\eps}\phi)|^2 \dx x_{1} \dx x_{2}+\frac{\lambda \eps}{4}\int_{\Omega_\eps}|\mathcal{U}^{-1}_{\eps}\phi|^4 \dx x_{1} \dx x_{2}-\frac{\mu_1}{\eps^2}\int_{\Omega_\eps}|\mathcal{U}^{-1}_{\eps}\phi|^2 \dx x_{1} \dx x_{2}.\nonumber
\end{align}
Notice that we have substracted the conserved quantity $\frac{\mu_1}{2\eps^{2}}\|\phi\|_{\sL^2}^2$ to the usual nonlinear energy, in order to deal with bounded energies. Indeed, we will consider initial conditions with bounded mass and energy, which means more precisely the following assumption.
\begin{assumption}
\label{assumption2}
There exists two constants $M_0>0$ and $M_1>0$ such that the initial data $\phi_0^\eps$ satisfies, for all $\eps\in (0,\eps_0)$,
$$\|\phi^\eps_0\|_{\sL^2}\leq M_0 \quad \mbox{and}\quad \mathcal{E}_{\eps}(\phi_0^\eps)\leq M_1.$$
\end{assumption}
Let us define the projection $\Pi_1$ on $e_1$ by letting $\Pi_1 u=\langle u, e_1\rangle_{\sL^2((-1,1))} e_1$. A consequence of Assumption \ref{assumption2} is that $\phi_0^\eps$ has a bounded $\sH^1$ norm and is close to its projection $\Pi_1\phi_{0}^\eps$. Indeed, we will prove the following lemma.
\begin{lemma}\label{cauchy-tensorial}
Assume that $\phi_0^\eps$ satisfies Assumption \ref{assumption2}. Then there exists $\eps_1(M_0)\in (0,\eps_0)$ and a constant $C>0$ independent of $\eps$ such that, for all $\eps\in(0,\eps_1(M_0))$,
\begin{equation}
\label{conf}
\|\phi_0^\eps\|_{\sH^1(\mathcal S)}\leq C\quad\mbox{and}\quad \|\phi_0^\eps-\Pi_1\phi_{0}^\eps\|_{\sL^2(\M,\sH^1(-1,1))}\leq C\eps.
\end{equation}
\end{lemma}

\subsection{Dimensional reduction and main results}
In the sequel, it will be convenient to consider the following change of temporal gauge $\phi^\eps(t; x_{1}, x_{2})=e^{-i\mu_{1}\eps^{-2} t}\phy^\eps(t; x_{1}, x_{2})$. This leads to the equation
\begin{equation}\label{CNLS''}
i\dr_{t}\phy^\eps=\mathcal{H}_{\eps}\phy^\eps+(V_{\eps}-\eps^{-2}\mu_{1})\phy^\eps+\lambda m_{\eps}^{-1} |\phy^\eps|^2\phy^\eps
\end{equation}
with conditions $\phy^\eps(t;x_{1},\pm 1)=0$, $\phy^\eps(0;\cdot)=\phi^\eps_{0}$.

\bigskip
In order to study \eqref{CNLS''}, the first natural try is to conjugate the equation by the unitary group $e^{it\mathcal{H}_{\eps}}$ so that the problem \eqref{CNLS''} becomes
\begin{equation}\label{CNLS''-conj}
i\dr_{t}\widetilde\phy^\eps=e^{it\mathcal{H}_{\eps}}(V_{\eps}-\eps^{-2}\mu_{1})e^{-it\mathcal{H}_{\eps}}\widetilde\phy^\eps+\lambda W_{\eps}(t;\widetilde\phy^\eps) ,\qquad \widetilde\phy^\eps(0;\cdot)=\phi_{0}^\eps,
\end{equation}
where
\begin{equation}\label{Weps}
W_{\eps}(t;\phy)=e^{it\mathcal{H}_{\eps}}m_{\eps}^{-1} |e^{-it\mathcal{H}_{\eps}}\phy|^2 e^{-it\mathcal{H}_{\eps}}\phy
\end{equation}
and where $\widetilde\phy^\eps=e^{it\mathcal{H}_{\eps}}\phy^\eps$ which satisfies $\widetilde\phy^\eps(t;x_{1},\pm 1)=0$. Nevertheless this reformulation does not take the inhomogeneity with respect to $\eps$ into account. In particular it will not be appropriate to estimate the approximation of the solution. 

We will see that \eqref{CNLS''} is well approximated in the limit $\eps\to 0$ by the following one dimensional equation
\begin{equation}\label{limit-equation}
i\dr_{t}\theta^\eps=D_{x_{1}}^2\theta^\eps-\frac{\kappa(x_{1})^2}{4}\theta^\eps+\lambda \gamma |\theta^\eps|^2\theta^\eps,
\end{equation}
with $\gamma=\int_{-1}^1 e_1(x_{2})^4 \dx x_{2}=3/4$ and $\theta^\eps(0,x_1)=\theta^\eps_0(x_1)=\int_{-1}^1\phi^\eps_0(x_1,x_2)e_1(x_2)\dx x_2$ for $x_1\in \M$ (recall that the notation $\M$ stands for $\R$ or $\T$). This last equation can be reformulated as
\begin{equation}\label{limit-equation-conj}
i\dr_{t}\widetilde\theta^\eps=e^{itD_{x_{1}}^2}\left(-\frac{\kappa^2(x_{1})}{4}\right)e^{-itD_{x_{1}}^2}\widetilde\theta^\eps+\lambda\gamma F(t;\widetilde\theta^\eps) ,\qquad \widetilde\theta(0;\cdot)=\theta^\eps_{0},
\end{equation}
where 
\begin{equation}\label{defF}
F(t; \widetilde\theta)=e^{itD_{x_{1}}^2}\left|e^{-itD_{x_{1}}^2}\widetilde\theta\right|^2e^{-itD_{x_{1}}^2}\widetilde\theta,
\end{equation}
and $\widetilde \theta^\eps=e^{itD_{x_{1}}^2}\theta^\eps$.

\bigskip
Our main results are the following theorems and their corollary.
\begin{theorem}[Solutions in the energy space]
\label{mainthmH1}
Consider a sequence of Cauchy data $\phi_0^\eps\in \sH^1_0(\mathcal S)$ satisfying Assumption \ref{assumption2}. Then:\\
(i) The limit problem \eqref{limit-equation} admits a unique solution $\theta^\eps\in C(\R_+;\sH^1(\M))\cap C^1(\R_+;\sH^{-1}(\M))$.\\[1mm]
(ii) There exists $\eps_1(M_0)\in (0,\eps_0]$ such that, for all $\eps\in (0,\eps_1(M_0))$, the two-dimensional problem \eqref{CNLS''} admits a unique solution $\phy^\eps\in C(\R_+;\sH^1_0(\mathcal{S}))\cap C^1(\R_+;\sH^{-1}(\mathcal{S}))$.\\[1mm]
(iii) For all $T>0$ there exists $C_T>0$ such that, for all $\eps\in (0,\eps_1(M_0))$, we have the error bound
\begin{equation}
\label{esti-erreur}
\sup_{t\in[0,T]}\|\phy^\eps(t)-\theta^\eps(t)e_1\|_{\sL^2(\mathcal{S})}\leq C_T\,\eps^{1/2}.
\end{equation}
\end{theorem}
\begin{theorem}[$\sH^2$ solutions]
\label{mainthmH2}
Assume that $\phi_0^\eps\in \sH^2\cap \sH^1_0(\mathcal S)$ and that there exist $M_0>0$, $M_2>0$ such that, for all $\eps\in (0,\eps_0)$, 
\begin{equation}
\label{ass3}
\|\phi_0^\eps\|_{\sL^2}\leq M_0,\qquad \left\|(\mathcal H_\eps-\frac{\mu_1}{\eps^2})\phi_0^\eps\right\|_{\sL^2}\leq M_2.
\end{equation}
Then $\phi_0^\eps$ satisfies Assumption \ref{assumption2} and:\\
(i) The limit problem \eqref{limit-equation} admits a unique solution $\theta^\eps\in C(\R_+;\sH^2(\M)\cap C^1(\R_+;\sL^2(\M))$.\\[1mm]
(ii) For all $\eps\in (0,\eps_1(M_0))$, the two-dimensional problem \eqref{CNLS''} admits a unique solution $\phy^\eps\in C(\R_+;\sH^2\cap \sH^1_0(\mathcal{S}))\cap C^1(\R_+;\sL^2(\mathcal{S}))$. The constant $\eps_1(M_0)$ is the same as in Theorem \ref{mainthmH1}.\\[1mm]
(iii) For all $T>0$ there exists $C_T>0$ such that, for all $\eps\in (0,\eps_1(M_0))$, we have the refined error bound
\begin{equation}
\label{esti-erreur2}
\sup_{t\in[0,T]}\|\phy^\eps(t)-\theta^\eps(t)e_1\|_{\sL^2(\mathcal{S})}\leq C_T\,\eps.
\end{equation}
\end{theorem}
Coming back by $\mathcal{U}_{\eps}^{-1}$ to the original equation \eqref{CNLS}, an immediate consequence of our two theorems is the following corollary.
\begin{corollary}
Consider a sequence of Cauchy data $\psi_0^\eps\in \sH^1_0(\Omega_\eps)$ with bounded mass and energy:
$$\|\psi_0^\eps\|_{\sL^2(\Omega_\eps)}\leq M_0\quad \mbox{and}\quad \frac{1}{2}\|\psi_0^\eps\|_{\sH^1(\Omega_\eps)}^2+\frac{\lambda \eps}{4}\|\psi_0^\eps\|_{\sL^4(\Omega_\eps)}^4-\frac{\mu_1}{2\eps^2}\|\psi_0^\eps\|_{\sL^2(\Omega_\eps)}^2\leq M_1.$$ Then for all $\eps\in (0,\eps_1)$, the confined NLS equation \eqref{CNLS} with $\alpha=1$ admits a unique solution $\psi^\eps\in C(\R_+;\sH^1_0(\Omega_\eps))\cap C^1(\R_+;\sH^{-1}(\Omega_\eps))$ and, for all $T>0$, we have the estimate
$$\sup_{t\in[0,T]}\|\psi^\eps(t)-e^{-i\mu_1/\eps^2t}\,\mathcal{U}_{\eps}^{-1}\left(\theta^\eps(t)e_1\right)\|_{\sL^2(\Omega_\eps)}\leq C_T\,\eps^{1/2}.$$
If, additionally, $\psi_0^\eps\in \sH^2(\Omega_\eps)$ and $\|(-\Delta_{\Omega_{\eps}}^\Dir-\frac{\mu_1}{\eps^2})\phi_0^\eps\|_{\sL^2(\Omega_\eps)}$ is bounded uniformly with respect to $\eps$, then we have $\psi^\eps\in C(\R_+;\sH^2\cap \sH^1_0(\Omega_\eps))\cap C^1(\R_+;\sL^2(\Omega_\eps))$ with, for all $T>0$, the estimate
$$\sup_{t\in[0,T]}\|\psi^\eps(t)-e^{-i\mu_1/\eps^2t}\,\mathcal{U}_{\eps}^{-1}\left(\theta^\eps(t)e_1\right)\|_{\sL^2(\Omega_\eps)}\leq C_T\,\eps.$$
\end{corollary}

\bigskip
This paper is organized as follows. Section \ref{S:prel} is devoted to technical lemmas related to the well-posedness of our Cauchy problems and to energy estimates. Section \ref{S:wp} deals with the proof of well-posedness stated in Theorems \ref{mainthmH1} and \ref{mainthmH2}. In Section \ref{S:red} we establish the tensorial approximation announced in Theorems \ref{mainthmH1} and \ref{mainthmH2}.

\section{Preliminaries}\label{S:prel}
In this section, we give a few technical results that will be useful in the sequel. 
\subsection{Norm equivalences}
Let us first remark that $$\mathcal P_{\eps,1}=(1-\eps x_2 \kappa(x_1))^{-1}D_{x_1}-\frac{\eps x_2 \kappa'(x_1)}{2(1-\eps x_2 \kappa(x_1))}.$$ Hence, by Assumption \ref{assumption1}, there exists three positive constants $C_1$, $C_2$, $C_3$ such that, for all $\eps\in (0,\eps_0)$ and for all $u\in \sH^1_0(\mathcal{S})$,
\begin{equation}
\label{equivnorm2}
\left(1-C_1\eps\right)\|\partial_{x_1}u\|_{\sL^2}\leq \|\mathcal P_{\eps,1}u\|_{\sL^2}+C_2\eps \|u\|_{\sL^2}\leq (1+C_3\eps)\|\partial_{x_1}u\|_{\sL^2}+C_3\eps\|u\|_{\sL^2}.
\end{equation}
Furthermore,  the graph norm of $\mathcal{H}_\eps$ is equivalent to the $\sH^2$ norm for all $\eps\in (0,\eps_0)$, with constants depending on $\eps$. More precisely, we have the following result.
\begin{lemma}
There exist positive constants $C_4$ and $C_5$ such that, for all $\eps\in (0,\eps_0)$ and for all $u\in \sH^2\cap\sH^1_0(\mathcal{S})$,
\begin{align}
&C_4\left(\left\|D_{x_1}^2u\right\|_{\sL^2}+\frac{1}{\eps^2}\left\|\left(D_{x_2}^2-\mu_1\right)u\right\|_{\sL^2}+\|u\|_{\sL^2}\right)\leq  \label{equivnorm}\\
&\qquad \leq \left\|\left(\mathcal{H}_\eps-\frac{\mu_1}{\eps^2}\right) u\right\|_{\sL^2}+\|u\|_{\sL^2}\leq C_5\left(\left\|D_{x_1}^2u\right\|_{\sL^2}+\frac{1}{\eps^2}\left\|\left(D_{x_2}^2-\mu_1\right)u\right\|_{\sL^2}+\|u\|_{\sL^2}\right). \nonumber
\end{align}
\end{lemma}
\begin{proof}
To prove the left inequality in \eqref{equivnorm}, we use standard elliptic estimates. For $u\in \sH^2\cap\sH^1_0(\mathcal{S})$, we let
\begin{equation}
\label{f}
f=\left(\mathcal{H}_{\eps}-\frac{\mu_1}{\eps^2}\right)u=\mathcal{P}_{\eps,1}^2u+\eps^{-2}(D_{x_{2}}^2-\mu_{1})u\end{equation}
and taking the $\sL^2$ scalar product of $f$ with $D_{x_1}^2 u$, we get
$$\langle D_{x_{1}}\mathcal{P}_{\eps,1}^2u, D_{x_{1}}u\rangle_{\sL^2}+\eps^{-2} \|D_{x_1}\left(D^2_{x_2}-\mu_1\right)^{1/2} u\|_{\sL^2}^2\leq \|f\|_{\sL^2}\|D_{x_{1}}^2u\|_{\sL^2}.$$
Then we write
\begin{align*}
\langle D_{x_{1}}\mathcal{P}_{\eps,1}^2u, D_{x_{1}}u\rangle_{\sL^2}
&=\left\|\mathcal{P}_{\eps,1}D_{x_{1}}u\right\|^2_{\sL^2}+\langle [D_{x_{1}},\mathcal{P}_{\eps,1}]u, \mathcal{P}_{\eps,1}D_{x_{1}}u\rangle_{\sL^2}\\
&\quad -\langle\mathcal{P}_{\eps,1}u,  [D_{x_{1}},\mathcal{P}_{\eps,1}] D_{x_{1}}u\rangle_{\sL^2}
\end{align*}
and use
\begin{equation}
\label{commutateur}
\left\|[D_{x_{1}},\mathcal{P}_{\eps,1}]u\right\|_{\sL^2}\leq C\eps \left(\|D_{x_1}u\|_{\sL^2}+ \|u\|_{\sL^2}\right),
\end{equation}
together with \eqref{equivnorm2} and the interpolation estimate $\|D_{x_1}u\|_{\sL^2}\leq C\|D_{x_1}^2u\|_{\sL^2}^{1/2}\|u\|_{\sL^2}^{1/2}$, to get
$$\langle D_{x_{1}}\mathcal{P}_{\eps,1}^2u, D_{x_{1}}u\rangle_{\sL^2}\geq (1-C\eps)\|D_{x_{1}}^2u\|_{\sL^2}^2-C\eps\|u\|_{\sL^2}^2.$$
It follows that
$$\|D_{x_{1}}^2u\|_{\sL^2}\leq C\|f\|_{\sL^2}+C\|u\|_{\sL^2}$$
and then, using \eqref{f} and again \eqref{equivnorm2},
\begin{align*}
\eps^{-2}\left\|(D_{x_{2}}^2-\mu_{1})u\right\|_{\sL^2}\leq \|f\|_{\sL^2}+\|\mathcal{P}_{\eps,1}^2u\|_{\sL^2}&\leq \|f\|_{\sL^2}+C\|D_{x_1}^2u\|_{\sL^2}+C\|u\|_{\sL^2}\\
&\leq C\|f\|_{\sL^2}+C\|u\|_{\sL^2}.
\end{align*}
This proves the left inequality in \eqref{equivnorm}. The right inequality can be easily obtained by using Minkowski inequality, \eqref{equivnorm2} and \eqref{commutateur}.
\end{proof}

In the sequel, we shall denote by $C$ a generic constant independent of $\eps$, and by $C_\eps$ a generic constant that depends on $\eps$. Moreover, for two positive numbers $\alpha^\eps$ and $\beta^\eps$, the notation $\alpha^\eps\lesssim \beta^\eps$ means that there exists $C>0$ {\em independent of $\eps$} such that for all $\eps\in (0, \eps_0)$, $\alpha^\eps\leq C\beta^\eps$.

\subsection{Some estimates of $F$ and $W_{\eps}$}
In this subsection, we give some results concerning the two nonlinear functions $F$ and $W_{\eps}$ defined in \eqref{defF} and \eqref{Weps}.
\begin{lemma}
\label{lemF}
The function $F$ is locally Lipschitz continuous on $\sH^1(\M)$ and on $\sH^2(\M)$: for $k=1$ or $k=2$,
\begin{equation}
\label{lip1}
\forall u_{1},u_{2}\in \sH^k(\M),\quad\sup_{t\in \M}\|F(t;u_{1})-F(t;u_{2})\|_{\sH^k}\lesssim (\|u_{1}\|_{\sH^k}^2+\|u_{2}\|_{\sH^k}^2)\|u_{1}-u_{2}\|_{\sH^k}
\end{equation}
and, for all $\eps\in (0, \eps_0)$, the function $W_{\eps}$ is locally Lipschitz continuous on $\sH^2\cap \sH^1_{0}(\mathcal S)$: there exists $C_\eps>0$ such that 
\begin{equation}
\label{lip-Weps}
\forall u_{1},u_{2}\in \sH^2\cap \sH^1_0(\mathcal S),\quad \sup_{t\in \M}\|W_{\eps}(t;u_{1})-W_{\eps}(t;u_{2})\|_{\sH^2}\leq C_\eps (\|u_{1}\|_{\sH^2}^2+\|u_{2}\|_{\sH^2}^2)\|u_{1}-u_{2}\|_{\sH^2}.
\end{equation}
Moreover, for all $u\in \sH^2(\M)$, one has
\begin{equation}
\label{boundF}
\sup_{t\in \R}\|F(t;u)\|_{\sH^2}\lesssim \|u\|_{\sH^1}^2\|u\|_{\sH^2}.
\end{equation}
Moreover, for all $M>0$ and for all $\eps\in (0,\eps_0)$, there exists a constant $C_\eps(M)>0$ such that, for all $u\in \sH^2\cap \sH^1_0(\mathcal S)$ with $\|u\|_{\sH^1}\leq M$, one has
\begin{equation}
\label{bg}
\sup_{t\in \R}\|W_{\eps}(t;u)\|_{\sH^2}\leq C_\eps(M)\big(1+\log \left(1+\|u\|_{\sH^2}\right)\big)\|u\|_{\sH^2}.
\end{equation}
\end{lemma}
\begin{proof}
We recall that the group $e^{-i\tau D^2_{x_{1}}}$ is unitary in $\sL^2(\M)$, $\sH^1(\M)$, $\sH^2(\M)$. Moreover, the group $e^{-i\tau \mathcal{H}_{\eps}}$  is unitary on $\sL^2(\mathcal S)$,  $\sH^1_{0}(\mathcal S)$ and $\sH^2(\mathcal S)\cap \sH^1_{0}(\mathcal S)$, if these two last spaces are respectively equipped with the norms $\|(\mathcal H_\eps u)^{1/2}\|_{\sL^2}$ and $\|\mathcal H_\eps u\|_{\sL^2}$, which are equivalent to the $\sH^1$ and $\sH^2$ norms with $\eps$-dependent constants, by \eqref{equivnorm}.

Let us prove \eqref{lip1}. We let $v_{j}=e^{-itD_{x_{1}}^2}u_{j}$. We have 
\begin{equation}\label{lipF}
 e^{-itD_{x_{1}}^2}(F(t ; u_{1})-F(t ; u_{2}))=|v_{1}|^2 v_{1}-|v_{2}|^2 v_{2}=(|v_{2}|^2+v_{1}\overline{v_{2}})(v_{1}-v_{2})+v_{1}^2(\overline{v_{1}}-\overline{v_{2}}).
\end{equation}
Then we have 
$$\|F(t ; u_{1})-F(t ; u_{2})\|_{\sH^k}\leq \||v_{2}|^2(v_{1}-v_{2})\|_{\sH^k}+\|v_{1}\overline{v_{2}}(v_{1}-v_{2})\|_{\sH^k}+\|\overline{v_{1}}^2(v_{1}-v_{2})\|_{\sH^k}.$$
We are led to estimate products of functions in $\sH^k$ in the form $v_{1} v_{2} v_{3}$ so that, by using the Sobolev embedding $\sH^1(\M)\hookrightarrow \sL^{\infty}(\M)$, we get for all $k\geq 1$
$$\|v_{1} v_{2} v_{3}\|_{\sH^k}\lesssim \|v_{1}\|_{\sH^k}\|v_{2}\|_{\sH^k}\|v_{3}\|_{\sH^k}.$$
Let us deal with \eqref{lip-Weps}. Here we let $v_{j}=e^{-it\mathcal{H}_{\eps}}u_{j}$ and we estimate
$$\|W_{\eps}(t; u_{1})-W_{\eps}(t; u_{2})\|_{\sH^2}\leq C_\eps \|m_{\eps}^{-1}(|v_{1}|^2v_{1}-|v_{2}|^2v_{2})\|_{\sH^2}\leq C_\eps' \||v_{1}|^2v_{1}-|v_{2}|^2v_{2}\|_{\sH^2}$$
where we have used the unitarity of $e^{-it\mathcal{H}_{\eps}}$ for the graph norm of $\mathcal{H}_{\eps}$. Then, the conclusion follows by using the embeddings $\sH^2(\mathcal S)\hookrightarrow \sL^{\infty}(\mathcal S)$ and $\sH^2(\mathcal{S})\hookrightarrow \sW^{1,4}(\mathcal{S})$. Let us now deal with \eqref{boundF}. We notice that, for all $u\in \sH^2(\M)$,
$$\|F(t; u)\|_{\sH^2}\lesssim \||v|^2v\|_{\sH^2}, \quad v=e^{-itD_{x_{1}}^2}u$$
and 
\begin{align*}
\||v|^2v\|_{\sH^2}&\lesssim \||v|^2v\|_{\sL^2}+\|v'^2 v\|_{\sL^2}+\|v''v^2\|_{\sL^2}\\
                                    &\lesssim \|v\|_{\sH^2}\|v\|^2_{\sH^1}+\|v'\|_{\sL^2}\|v'\|_{\sL^\infty}\|v\|_{\sL^\infty}+\|v\|^2_{\sL^\infty}\|v\|_{\sH^2}\\
                                    &\lesssim \|v\|_{\sH^1}^2\|v\|_{\sH^2}= \|u\|_{\sH^1}^2\|u\|_{\sH^2}.
\end{align*}
Let us now deal with \eqref{bg}. We first recall the Gagliardo-Nirenberg inequality in dimension 2 (see \cite[p. 129]{Nir59}):
\begin{equation}\label{GN}
\|v\|_{\sW^{1,4}}^2\lesssim\|v\|_{\sL^\infty}\|v\|_{\sH^2}.
\end{equation}
The next Sobolev inequality is due to Br\'ezis and Gallouet (see \cite[Lemma 2]{BG80}): there exists $C(M)>0$ such that, for all $v\in \sH^2(\R^2)$ with $\|v\|_{\sH^1(\R^2)}\leq M$,
\begin{equation}
\label{soboBG}
\|v\|_{\sL^\infty}\leq C(M)\left(1+\sqrt{\log(1+\|v\|_{\sH^2})}\right).
\end{equation}
By using continuous extensions from $\sH^2(\mathcal S)$ to $\sH^2(\R^2)$, one obtains the same inequality for $u\in \sH^2\cap \sH^1_0(\mathcal S)$. Hence, for all $v\in \sH^2(\mathcal S)$ with $\|v\|_{\sH^1}\leq M$,
\begin{align*}
\||v|^2v\|_{\sH^2}\lesssim \|v^3\|_{\sL^2}+\|\Delta (v^3)\|_{\sL^2}&\lesssim \|v\|_{\sL^6}^3+\|v^2\Delta v\|_{\sL^2}+\|v |\nabla v|^2\|_{\sL^2}\\
&\lesssim \|v\|_{\sH^1}^3+ \|v\|_{\sL^\infty}^2\|\Delta v\|_{\sL^2}+\|v\|_{L^\infty}\|v\|_{\sW^{1,4}}^2\\
&\lesssim C(M)\left(1+\log(1+\|v\|_{\sH^2})\right)\|v\|_{\sH^2},
\end{align*}
where we used the Sobolev embedding $\sH^1(\mathcal S)\hookrightarrow \sL^6(\mathcal S)$, \eqref{GN} and \eqref{soboBG}. Finally, for all $u\in \sH^2\cap \sH^1_0 (\mathcal S)$ with $\|u\|_{\sH^1}\leq M$, setting $v=e^{-it\mathcal{H}_{\eps}}u$ we get $\|v\|_{\sH^1}\leq C_\eps M$ and 
\begin{align*}
\|W_{\eps}(t;u)\|_{\sH^2}\leq C_\eps \||v|^2v\|_{\sH^2}&\leq C_\eps(M)\left(1+\log(1+\|v\|_{\sH^2})\right)\|v\|_{\sH^2}\\
                                                                                             &\leq C_\eps' (M)\left(1+\log(1+\|u\|_{\sH^2})\right)\|u\|_{\sH^2}.
\end{align*}
This proves \eqref{bg} and the proof of the lemma is complete.
\end{proof}

\subsection{Proof of Lemma \ref{cauchy-tensorial}}
We will need the following easy lemma.
\begin{lemma}\label{Sob-anis}
For all $u\in\sH^1(\M)$, we have
\begin{equation}
\label{sobo1D}\|u\|_{\sL^4}^4\leq 2\|u\|_{\sL^2}^3\|u'\|_{\sL^2}.
\end{equation}
For all $u\in\sH^1(\mathcal{S})$, we have
\begin{equation}
\label{sobo2D}
\|u\|_{\sL^4}^4\leq 4\|u\|^2_{\sL^2(\mathcal{S})} \|\partial_{x_{1}}u\|_{\sL^2(\mathcal{S})} \|\partial_{x_{2}}u\|_{\sL^2(\mathcal{S})}.
\end{equation}
\end{lemma}
\begin{proof}
The proof of \eqref{sobo1D} is a consequence of the standard inequality, for $f\in\sH^1(\M)$, $\|f\|^2_{\sL^\infty}\leq 2\|f\|_{\sL^2}\|f'\|_{\sL^2}$.
To prove \eqref{sobo2D}, let us recall the following inequality
$$\int_{\mathcal{S}} |f|^2 \dx x_{1} \dx x_{2}\leq \|\dr_{x_{1}}f\|_{\sL^1(\mathcal{S})} \|\dr_{x_{2}}f\|_{\sL^1(\mathcal{S})},\quad\forall f\in\sW^{1,1}(\mathcal{S}).$$
Indeed, by density and extension, we may assume that $f\in\mathcal{C}^\infty_{0}(\R^2)$ and we can write
$$f(x_{1}, x_{2})=\int_{-\infty}^{x_{1}} \partial_{x_{1}}f(u,x_{2}) \dx u,\qquad f(x_{1}, x_{2})=\int_{-\infty}^{x_{2}} \partial_{x_{2}}f(x_{1},v) \dx v.$$
We get
$$|f(x_{1}, x_{2})|^2\leq \left(\int_{\R} |\partial_{x_{1}}f(u,x_{2})|\dx u\right)\left(\int_{-1}^1 |\partial_{x_{2}}f(x_{1},v)|\dx v\right)$$
and it remains to integrate with respect to $x_{1}$ and $x_{2}$.
We apply this inequality to $f=u^2$, use the Cauchy-Schwarz inequality and \eqref{sobo2D} follows.
\end{proof}
Now, we prove a technical lemma on the energy functional.
\begin{lemma}
\label{lemma-energy}
There exists $\eps_2\in (0,\eps_0)$ such that, for all $\eps\in (0,\eps_2)$, the energy functional defined by \eqref{Energy} satisfies the following estimate. For all $M>0$, there exists $C_0>0$ such that, for all $\phy\in \sH^1_0(\mathcal S)$ with $\|\phy\|_{\sL^2}\leq M$, one has
\begin{equation}
\label{minorE}
\mathcal{E}_{\eps}(\phy)\geq \frac{1}{4}\|\partial_{x_1}\phy\|_{\sL^2(\mathcal S)}^2+\left(\frac{3}{8\eps^2}-C_0M^4\right)\|\partial_{x_2}(\mathsf{Id}-\Pi)\phy\|^2_{\sL^2(\mathcal S)}-C_0M^2-C_0M^6.
\end{equation}
\end{lemma}
\begin{proof}
Remark that
\begin{align*}
\mathcal{E}_{\eps}(\phy)&=\frac{1}{2}\int_{\mathcal{S}}|\mathcal{P}_{\eps,1}\phy|^2 \dx x_{1} \dx x_{2}+\frac{1}{2\eps^2}\left \langle \left(D_{x_2}^2-\mu_1\right)\phy, \phy\right\rangle_{\sL^2}+\frac{1}{2}\int_{\mathcal{S}}V_{\eps}|\phy|^2 \dx x_{1} \dx x_{2}\\
&+\frac{\lambda}{4}\int_{\mathcal{S}}m_{\eps}^{-1}|\phy|^4 \dx x_{1} \dx x_{2}.
\end{align*}
Next, recalling that $\Pi_1$ denotes the projection on the first eigenfunction $e_1$ of $D_{x_2}^2$, we easily get
$$\|\phy\|_{\sL^4(\mathcal{S})}^4\leq 8\|\Pi_1 \phy\|^4_{\sL^4(\mathcal{S})}+8\|(\mathsf{Id}-\Pi_1) \phy\|^4_{\sL^4(\mathcal{S})}.$$
We may write $\Pi_1\phy(x_{1}, x_{2})=\theta(x_{1}) e_1(x_{2})$ so that, with \eqref{sobo1D},
\begin{align}
\|\Pi_1\phy\|_{\sL^4(\mathcal{S})}^4&=\gamma \int_{\M} \theta(x_{1})^4 \dx x_{1}\leq 2\gamma \|\theta\|_{\sL^2(\M)}^3\|\theta'\|_{\sL^2(\M)}=2\gamma \|\Pi_1\phy\|_{\sL^2(\mathcal{S})}^3\|\partial_{x_{1}}(\Pi_1\phy)\|_{\sL^2(\mathcal{S})}\nonumber\\
&\leq 2\gamma\|\phy\|_{\sL^2(\mathcal{S})}^3\|\partial_{x_{1}}(\Pi_1\phy)\|_{\sL^2(\mathcal{S})}\label{normL4}
\end{align}
where $\gamma=\int_{-1}^1 e_1(x_{2})^4 \dx x_{2}$, and thus, for all $\eta\in(0,1)$,
$$\|\Pi_1\phy\|_{\sL^4(\mathcal{S})}^4\leq \eta \|\Pi_1 \partial_{x_{1}}\phy\|_{\sL^2(\mathcal{S})}^2+\eta^{-1}\gamma^2\|\phy\|_{\sL^2(\mathcal{S})}^6.$$
Moreover, thanks to \eqref{sobo2D}, we have, for all $\eta\in(0,1)$,
\begin{align}
\|(\mathsf{Id}-\Pi_1) \phy\|^4_{\sL^4(\mathcal{S})}&\leq 4\|\phy\|^2_{\sL^2(\mathcal{S})} \|\partial_{x_{1}}(\mathsf{Id}-\Pi_1)\phy\|_{\sL^2(\mathcal{S})} \|\partial_{x_{2}}(\mathsf{Id}-\Pi_1)\phy\|_{\sL^2(\mathcal{S})}\nonumber\\
&\leq \eta\|\partial_{x_{1}}(\mathsf{Id}-\Pi_1)\phy\|^2_{\sL^2(\mathcal{S})}+4 \eta^{-1}\|\phy\|^4_{\sL^2(\mathcal{S})}\|\partial_{x_{2}}(\mathsf{Id}-\Pi_1)\phy\|^2_{\sL^2(\mathcal{S})}.\label{lun1}
\end{align}
Now we remark that, if $\mu_2=\pi^2$ denotes the second eigenvalue of $D_{x_2}^2$ on $(-1,1)$ with Dirichlet boundary conditions, we have
\begin{equation}\label{spectral-gap}
\left \langle \left(D_{x_2}^2-\mu_1\right)\phy, \phy\right\rangle_{\sL^2(\mathcal S)}\geq \left(1-\frac{\mu_1}{\mu_2}\right)\left\|\partial_{x_2}(\mathsf{Id}-\Pi_1)\phy\right\|_{\sL^2(\mathcal S)}^2=\frac{3}{4}\left\|\partial_{x_2}(\mathsf{Id}-\Pi_1)\phy\right\|_{\sL^2(\mathcal S)}^2.
\end{equation}
Therefore, using \eqref{equivnorm2}, \eqref{lun1}, \eqref{spectral-gap}, using that $\|V_\eps\|_{\sL^\infty}\leq C$ and that $0\leq m_\eps^{-1}\leq 1+C\eps$, we obtain
\begin{align*}
\mathcal{E}_{\eps}(\phy)&\geq \frac{1}{2}(1-C\eps)\|\partial_{x_1}\phy\|_{\sL^2(\mathcal S)}^2-C\|\phy\|_{\sL^2(\mathcal S)}^2+\frac{3}{8\eps^2}\left\|\partial_{x_2}(\mathsf{Id}-\Pi_1)\phy\right\|_{\sL^2(\mathcal S)}^2\\
&\quad -2|\lambda|(1+C\eps)\left(\eta \|\partial_{x_1}\phy\|_{\sL^2(\mathcal S)}^2+4 \eta^{-1}\|\phy\|^4_{\sL^2(\mathcal{S})}\left\|\partial_{x_2}(\mathsf{Id}-\Pi_1)\phy\right\|_{\sL^2(\mathcal S)}^2\right)-C\|\phy\|_{\sL^2(\mathcal S)}^6\\
&\geq \frac{1}{4}\|\partial_{x_1}\phy\|_{\sL^2(\mathcal S)}^2+\left(\frac{3}{8\eps^2}-C\|\phy\|^4_{\sL^2(\mathcal{S})}\right)\left\|\partial_{x_2}(\mathsf{Id}-\Pi_1)\phy\right\|_{\sL^2(\mathcal S)}^2-C\|\phy\|_{\sL^2(\mathcal S)}^2-C\|\phy\|_{\sL^2(\mathcal S)}^6
\end{align*}
where we has chosen $\eta=\frac{1-2C\eps}{8|\lambda|(1+C\eps)}$, which is positive for $\eps$ small enough.
\end{proof}

\begin{proofof}{Lemma \ref{cauchy-tensorial}}
It is easy now to deduce Lemma \ref{cauchy-tensorial} from Lemma \ref{lemma-energy}. Indeed, consider a sequence $\phi_0^\eps$ satisfying Assumption \ref{assumption2} and introduce the constants
\begin{equation}
\label{eps1}
\eps_1(M_0)=\min\left(\eps_2,\left(\frac{3}{16C_0M_0^4}\right)^{1/2}\right).
\end{equation}
We deduce from \eqref{minorE} that, if $\eps\in (0,\eps_1(M_0))$, we have
\begin{align}
\frac{3}{16}\left(\|\partial_{x_1}\phi_0^\eps\|_{\sL^2}^2+\frac{1}{\eps^2}\|\partial_{x_2}(\mathsf{Id}-\Pi_{1})\phi_0^\eps\|^2_{\sL^2}\right)&\leq  \frac{1}{4}\|\partial_{x_1}\phi_0^\eps\|_{\sL^2}^2+\left(\frac{3}{8\eps^2}-C_0M^4\right)\|\partial_{x_2}(\mathsf{Id}-\Pi_{1})\phi_0^\eps\|^2_{\sL^2}\nonumber\\
&\leq \mathcal{E}_{\eps}(\phi_0^\eps)+C_0M_0^2+C_0M_0^6\nonumber\\
&\leq M_1+C_0M_0^2+C_0M_0^6.\label{i1}
\end{align}
The conclusion \eqref{conf} stems from \eqref{i1} by remarking also that
\begin{equation*}
\|\partial_{x_2}\Pi_{1} \phi_0^\eps\|_{\sL^2}=\|\langle \phi_0^\eps, e_1\rangle_{\sL^2((-1,1))} \partial_{x_2}e_1\|_{\sL^2}\leq \frac{\pi}{2}\|\phi_0^\eps\|_{\sL^2}\leq \frac{\pi}{2}M_0
\end{equation*}
and by using the Poincar\'e inequality
\begin{equation*}
\|(\mathsf{Id}-\Pi_{1})\phi_0^\eps\|_{\sL^2(\M,\sH^1(-1,1))}\leq \frac{\sqrt{1+\pi^2}}{\pi}\|\partial_{x_2}(\mathsf{Id}-\Pi_{1})\phi_0^\eps\|_{\sL^2}.
\end{equation*}
\end{proofof}

\section{Well-posedness of the Cauchy problems}\label{S:wp}

\subsection{Limit equation}
The aim of this subsection is to prove briefly the global well-posedness of the limit equation \eqref{limit-equation}. 
\begin{proposition}\label{wp-limit}
Let $\theta_0\in \sH^1(\M)$. Then \eqref{limit-equation} with the Cauchy data $\theta_0$ admits a unique global solution $\theta\in C(\R_+;\sH^1(\M))\cap C^1(\R_+;\sH^{-1}(\M))$, that satisfies the following conservation laws
\begin{align}
\|\theta(t;\cdot)\|_{\sL^2}&=\|\theta_{0}\|_{\sL^2}\quad \mbox{(mass)},\label{consmass}\\
E(\theta(t;\cdot))&=E(\theta_{0})\quad \mbox{(nonlinear energy)},\label{consenergy}
\end{align}
where
$$E(\theta)=\frac{1}{2}\int_{\M}\left( |\dr_{x_{1}}\theta|^2-\frac{\kappa(x_{1})^2}{4}|\theta|^2\right)\dx x_{1}+\lambda \frac{\gamma}{4}\int_{\M}\left|\theta\right|^4 \dx x_{1}.$$
Moreover, there exists a constant $C>0$ such that
\begin{equation}\label{boundH1-theta}
\forall t\in [0,\R_+),\qquad \|\theta(t)\|_{\sH^1}\leq C(\|\theta_0\|_{\sH^1}+\|\theta_0\|_{\sH^1}^2)
\end{equation}
If $\theta_0\in \sH^2(\M)$ then $\theta\in C(\R_+;\sH^2(\M))\cap C^1(\R_+;\sL^2(\M))$ and
\begin{equation}\label{boundH2-theta}
\forall t\in \R_+,\qquad \|\theta(t)\|_{\sH^2}\leq \|\theta_0\|_{\sH^2}\exp(C(1+\|\theta_0\|_{\sH^1}^4)t).
\end{equation}
\end{proposition}
\begin{proof}
We introduce
$$\mathcal{F}(\theta)(t)=\theta_{0}-i\int_{0}^t \left\{e^{isD_{x_{1}}^2}\left(-\frac{\kappa^2(x_{1})}{4}\right)e^{-isD_{x_{1}}^2}\theta(s;\cdot)+\lambda\gamma F(s;\theta(s ;\cdot))\right\} \dx s.$$
For $M>0, T>0$, we consider the complete space
$$G_{T, M}=\{C([0,T] ; \sH^1(\M)) :  \forall t\in[0,T],\quad \theta(t)\in \overline{\mathcal{B}}_{\sH^1}(\theta_{0},M)\},$$
where, for all Banach space $X$, $\overline{\mathcal{B}}_{X}(\theta_{0},M)$ denotes the closed ball in $X$, of radius $M$, centered in $\theta_0$. Let us briefly explain why $\mathcal{F}$ is a contraction from $G_{T, M}$ to $G_{T, M}$ as soon as $T$ is small enough. Due to \eqref{lip1} and $F(t,0)=0$, there exists $C>0$ such that for all $M, T>0, t\in[0,T]$ and $\theta\in G_{T,M}$, we have
$$\|\mathcal{F}(\theta)(t)-\theta_{0}\|_{\sH^1}\leq CT+CM^3T$$
which leads to choose $T\leq T_1=M(C+CM^3)^{-1}$. In the same way, there exists $C>0$ such that for all $M, T>0, t\in[0,T]$ and $u_{1}, u_{2}\in G_{T,M}$,
$$\|\mathcal{F}(\theta_{1})(t)-\mathcal{F}(\theta_{2})(t)\|_{\sH^1}\leq (CT+CM^2T) \sup_{t\in[0,T]}\|\theta_{1}(t)-\theta_{2}(t)\|_{\sH^1}$$
so that we choose $T<T_2= (C+CM^2)^{-1}$. It remains to apply the fixed point theorem for any $T\in\left(0,(\min(T_1,T_2)\right)$ and the conclusion is standard. By a continuation argument, it is clear moreover that the solution is global in time if it is bounded in $\sH^1$.

\bigskip
The conservation of the $\sL^2$-norm \eqref{consmass} is obtained by considering the scalar product of \eqref{limit-equation} with $\theta$ and then taking the imaginary part. For the conservation of the energy \eqref{consenergy}, we consider the scalar product of the equation with $\partial_{t}\theta$ and take the real part. Let us now prove \eqref{boundH1-theta}. If $\lambda\geq 0$, it is an immediate consequence of the bounds on the energy and $\sL^2$-norm and the Sobolev embedding $\sH^1(\M)\hookrightarrow \sL^4(\M)$. Let us analyze the case $\lambda<0$. Thanks to \eqref{sobo1D}, we have
$$\frac{|\lambda\gamma|}{4}\int_{\M}|\theta|^4\dx x_{1}\leq \frac{|\lambda\gamma|}{2}\|\theta\|_{\sL^2(\M)}^3\|\partial_{x_{1}}\theta\|_{\sL^2(\M)}=  \frac{|\lambda\gamma|}{2}\|\theta_{0}\|_{\sL^2(\M)}^3\|\partial_{x_{1}}\theta\|_{\sL^2(\M)}$$
so that, for all $\eta\in(0,1)$,
$$\frac{|\lambda\gamma|}{4}\int_{\M}|\theta|^4\dx x_{1}\leq\frac{|\lambda\gamma|}{4}\left(\eta^{-1}\|\theta_{0}\|_{\sL^2(\M)}^{ 6}+\eta\|\partial_{x_{1}}\theta\|_{\sL^2(\M)}^2\right).$$
Choosing $\eta$ such that $\eta\frac{|\lambda\gamma|}{4}<\frac{1}{2}$ and using the bound on the energy, we get the uniform estimate \eqref{boundH1-theta}. In particular, the solution $\theta$ is global in time.

\bigskip
The local well-posedness in $\sH^2(\M)$ can be obtained by a similar procedure. To prove that the $\sH^2$ solution is global in time, we simply use Assumption \ref{assumption1} on $\kappa$ with \eqref{boundF}:
$$\|\theta(t)\|_{\sH^2}\leq \|\theta_0\|_{\sH^2}+\int_0^t C(1+\|\theta(s)\|_{\sH^1}^2)\|\theta(s)\|_{\sH^2}\dx s$$
and conclude by using the $\sH^1$ bound \eqref{boundH1-theta} and the Gronwall lemma.
\end{proof}

\subsection{Cauchy problem in the strip}
\label{cauchystrip}
Let us now analyze the well-posedness of \eqref{CNLS''}, but without any $\eps$-control of the solution.
\begin{proposition}\label{wp}
Let $\phi_0^\eps\in \sH^1_0(\mathcal S)$ and let $\eps\in (0,\eps_0)$. Then, the following properties hold:\\
(i) The problem \eqref{CNLS''} admits a unique maximal solution $\phy^\eps\in C([0,T_{\rm max}^\eps);\sH^1_0(\mathcal S))\cap C^1([0,T_{\rm max}^\eps);\sH^{-1}(\mathcal S))$, with $T_{\rm max}^\eps\in (0,+\infty]$ that satisfies the following conservation laws
\begin{align}
\|\phy^\eps(t;\cdot)\|_{\sL^2}&=\|\phi^\eps_{0}\|_{\sL^2}\quad \mbox{(mass)},\label{consmass2}\\
\mathcal E_\eps(\phy^\eps(t;\cdot))&=\mathcal E_\eps(\phi^\eps_{0})\quad \mbox{(nonlinear energy)},\label{consenergy2}
\end{align}
where $\mathcal E_\eps$ is defined in \eqref{Energy}.\\
(ii) There exists a constant $C_1>0$ such that, if $\eps<\eps_2$ (given in Lemma \ref{lemma-energy}) and if $\eps \|\phi_0^\eps\|_{\sL^2}^2\leq C_1$, then $T_{\rm max}^\eps=+\infty$.\\
(iii) If $\phi_0^\eps$ belongs to $\sH^2\cap \sH^1_0(\mathcal S)$, then $\phy^\eps\in C([0,T_{\rm max}^\eps);\sH^2\cap \sH^1_0(\mathcal S))\cap C^1([0,T_{\rm max}^\eps);\sL^2(\mathcal S))$. 
\end{proposition}
\begin{proof}
{\em Step 1: local well-posedness in $\sH^2$.} Let us fix $\eps\in(0,\eps_{0})$ and analyze in a first step the well-posedness in $\sH^2\cap\sH^1_0(\mathcal S)$. For $\phi^\eps_0\in  \sH^2\cap \sH^1_0(\mathcal S)$, we consider the conjugate problem \eqref{CNLS''-conj} in its Duhamel form
$$\widetilde \phy^\eps(t)=\phi^\eps_{0}-i\int_{0}^t \left(e^{is\mathcal{H}_{\eps}}(V_{\eps}-\eps^{-2}\mu_{1})e^{-is\mathcal{H}_{\eps}}\widetilde \phy^\eps(s)+\lambda W_{\eps}(s ; \widetilde \phy^\eps(s))\right)\dx s=\mathcal{W}_{\eps}(\widetilde \phy^\eps)(t).$$
For $M, T>0$, we consider the complete space
$$\widetilde G_{T, M}=\{C([0,T] ; \sH^2\cap \sH^1_0(\mathcal S)) :  \forall t\in[0,T],\quad \theta(t)\in \overline{\mathcal{B}}_{\sH^2}(\theta_{0},M)\}.$$ 
The application $\mathcal{W}_{\eps}$ is a contraction from $\widetilde G_{T,M}$ to $\widetilde G_{T,M}$ for $T$ small enough. Indeed, as in Lemma \ref{wp-limit} and thanks to \eqref{lip-Weps}, there exists $C_\eps>0$ such that for all $T>0$, $M>0$, $t\in[0,T]$ and $\phy_{1}, \phy_{2}\in \widetilde G_{T,M}$,
$$\|\mathcal{W}_{\eps}(\phy_{1})(t)-\phy_{0}\|_{\sH^2}\leq C_\eps T+C_\eps TM^3,$$
$$\|\mathcal{W}_{\eps}(\phy_{1})(t)-\mathcal{W}_{\eps}(\phy_{2})(t)\|_{\sH^2}\leq (C_\eps T+C_\eps TM^2)\sup_{t\in[0,T]}\|\phy_{1}(t)-\phy_{2}(t)\|_{\sH^2},$$
where we have again used the unitarity of $e^{it\mathcal{H}_{\eps}}$ with respect to the graph norm of $\mathcal{H}_{\eps}$ and the equivalence between the graph norm of $\mathcal{H}_{\eps}$ and the $\sH^2$-norm, for each fixed $\eps$. Therefore the Banach fixed point theorem insures the existence and uniqueness of a local in time solution of \eqref{CNLS''-conj} and thus of \eqref{CNLS''} for each given $\eps\in(0,\eps_{0})$. The conservation laws \eqref{consmass2} and \eqref{consenergy2} are obtained similarly as \eqref{consmass} and \eqref{consenergy}. In fact, it is not difficult to deduce the existence of a maximal existence time $T_{\rm max, \sH^2}^\eps\in (0,+\infty]$ such that $\phy^\eps\in C([0,T_{\rm max, \sH^2}^\eps);\sH^2\cap \sH^1_0(\mathcal S))\cap C^1([0, T_{\rm max, \sH^2}^\eps);\sL^2(\mathcal S))$ and such that we have the alternative
\begin{equation}
\label{alternative2}
T_{\rm max, \sH^2}^\eps=+\infty \quad \mbox{or}\quad \lim_{t\to T_{\rm max, \sH^2}^\eps}\|\phy^\eps(t)\|_{\sH^2}=+\infty.
\end{equation}

\bigskip
\noindent
{\em Step 2: local well-posedness in $\sH^1$.} Consider now a Cauchy data $\phi_0^\eps\in \sH^1_0(\mathcal S)$. To prove the local well-posedness in $\sH^1$, we can proceed with the usual argument based on Trudinger's inequality, explained in  Section 3.6 of \cite{Cazenave} and that we sketch here. 

We first recall the construction of  a local weak solution $\phy^\eps\in L^\infty([0,T);\sH^1_0(\mathcal S))\cap W^{1,\infty}((0,T);\sH^{-1}(\mathcal S))$ of  \eqref{CNLS''} by a standard regularization method. We approximate the Cauchy data $\phi_0^\eps$ by a sequence $\phi_{0,n}^{\eps}\in\sH^2\cap \sH^1_0(\mathcal S)$ converging to $\phi_0^\eps$ in $\sH^1_0(\mathcal S)$. Then we apply the well-posedness result in $\sH^2$ proved in Step 1 to obtain a sequence of solutions $\phy^\eps_n \in C([0,T_n];\sH^2\cap \sH^1_0(\mathcal S))\cap C^1([0,T_n];\sL^2(\mathcal S))$ of \eqref{CNLS''} with $\phy^\eps_n(0;\cdot)=\phi^\eps_{0,n}$, satisfying the conservation of mass and energy and where $T_n$ is chosen such that
$$\forall t\in [0,T_n], \qquad \|\phy^\eps_n(t)\|_{\sH^1}\leq 2\|\phi_0^\eps\|_{\sH^1}.$$
From \eqref{CNLS''-conj} and the embedding $\sH^1(\mathcal S)\hookrightarrow \sL^6(\mathcal S)$, we deduce that for $s,t\leq T_n$
$$\|\phy^\eps_n(t)-\phy^\eps_n(s)\|_{\sL^2}=\|\widetilde\phy^\eps_n(t)-\widetilde \phy^\eps_n(s)\|_{\sL^2}\leq C_\eps |t-s|\left(\|\phi_0^\eps\|_{\sH^1}+\|\phi_0^\eps\|_{\sH^1}^3\right)$$
and then, from the conservation of mass and energy, we get
\begin{align*}
\|\phy^\eps_n(t)\|_{\sH^1}^2&\leq \|\phi_{0,n}^{\eps}\|_{\sL^2}^2+\|\nabla \phi_{0,n}^{\eps}\|_{\sL^2}^2+\left\|V_\eps-\frac{\mu_1}{\eps^2}\right\|_{\sL^\infty}\left|\|\phy^\eps_n(t)\|_{\sL^2}^2-\|\phy^\eps_n(0)\|_{\sL^2}^2\right|\\
&\quad +\frac{|\lambda|}{2}\|m^\eps\|_{\sL^\infty}\left|\|\phy^\eps_n(t)\|_{\sL^4}^4-\|\phy^\eps_n(0)\|_{\sL^4}^4\right|\\
&\leq \|\phi^\eps_{0,n}(t)\|_{\sH^1}^2+C_\eps t\left(\|\phi_0^\eps\|_{\sH^1}+\|\phi_0^\eps\|_{\sH^1}^3\right)^2.
\end{align*}
From this estimate, we deduce that there exists $T>0$, independent of $n$ (but of course depending on $\eps$), such $T_n\geq T$. The sequence $(\phy^\eps_n)_{n\in \mathbb N}$ being bounded in $L^\infty([0,T);\sH^1_0(\mathcal S))\cap W^{1,\infty}((0,T);\sH^{-1}(\mathcal S))$, we can use the local compactness of $\sH^1$ into $\sL^6$ to extract a subsequence that converges to a weak solution of \eqref{CNLS'}. This weak solution satisfies in fact $\|\phy^\eps(t)\|_{\sL^2}=\|\phi^\eps_0\|_{\sL^2}$ and the inequality $\mathcal E_\eps(\phy^\eps(t))\leq \mathcal E_\eps(\phi^\eps_0)$.

Next, by using Ogawa's method \cite{Ogawa} (see Theorem 3.6.1 in \cite{Cazenave}), we prove the uniqueness of the weak solution $\phy^\eps\in L^\infty([0,T);\sH^1_0(\mathcal S))\cap W^{1,\infty}([0,T);\sH^{-1}(\mathcal S))$. This crucial property relies on an $\sL^2$ estimate and Trudinger's inequality.

A consequence of the uniqueness property is that the NLS equation \eqref{CNLS'} is time-reversible, so one has $\mathcal E_\eps(\phi^\eps_0)\leq \mathcal E_\eps(\phy^\eps(t))$ and then the energy is exactly conserved: the weak solution $\phy^\eps$ satisfies \eqref{consenergy2}. From this, one deduces (see Theorem 3.3.9 of \cite{Cazenave}) that $\phy^\eps\in C([0,T];\sH^1_0(\mathcal S))\cap C^1([0,T);\sH^{-1}(\mathcal S))$, that the solution depends continuously from the initial data, and that the exists a maximal existence time $T_{\rm \max, \sH^1}^\eps\in (0,+\infty]$ with the alternative
\begin{equation}
\label{alternative}
T_{\rm max, \sH^1}^\eps=+\infty \quad \mbox{or}\quad \lim_{t\to T_{\rm max, \sH^1}^\eps}\|\phy^\eps(t)\|_{\sH^1}=+\infty.
\end{equation}

\bigskip
\noindent
{\em Step 3: equality of the maximal existence times}. Let $\phi_0^\eps\in \sH^2\cap\sH^1_0(\mathcal S)$. From the previous two steps, there exists a maximal existence time $T_{\rm max, \sH^2}^\eps$ (resp. $T_{\rm max, \sH^1}^\eps$) of the $\sH^2$ (resp. $\sH^1$) solution of \eqref{CNLS'}. Moreover, by \eqref{alternative2} and \eqref{alternative}, it is already obvious that $T_{\rm max, \sH^2}^\eps\leq T_{\rm max, \sH^1}^\eps$. Let us prove by a contradiction argument that we have in fact the equality of these two maximal existence times:
\begin{equation}
\label{equa}
T_{\rm max, \sH^2}^\eps=T_{\rm max, \sH^1}^\eps=T^\eps_{\rm max}.
\end{equation}
Assume that $T_{\rm max, \sH^2}^\eps< T_{\rm max, \sH^1}^\eps$. Then $\phy^\eps$ is bounded by a constant $M^\eps$ in $\sH^1$ norm on $[0,T_{\max,\sH^2}^\eps]$ and one has
\begin{equation}
\label{liminfini}
 \lim_{t\to T_{\rm max, \sH^2}^\eps}\|\phy^\eps(t)\|_{\sH^2}=+\infty.
\end{equation}
From \eqref{CNLS''} and \eqref{bg} we get
$$\|\dr_{t}\phy^\eps\|_{\sH^2}\leq C_\eps\big(1+\log\left(1+\|\phy^\eps(t;\cdot)\|_{\sH^2}\right)\big)\|\phy^\eps(t;\cdot)\|_{\sH^2}.$$
It remains to use an argument \textit{\`a la} Gronwall from \cite{BG80}. Given a Banach space $G$, let us consider a function $\phy\in\mathcal{C}^1([0,T^*),G)$ such that for, $t\in[0,T^*)$,
$$\|\phy'(t)\|\leq C(1+\log(1+\|\phy(t)\|))\|\phy(t)\|.$$
We easily get
$$\|\phy(t)\|\leq F(t),\qquad \mbox{ with }\quad F(t)=\|\phy_{0}\|+C\int_{0}^t (1+\log(1+\|\phy(\tau)\|))\|\phy(\tau)\|\dx\tau$$
and
$$\frac{\dx}{\dx t}F(t)=C(1+\log(1+\|\phy(t)\|))\|\phy(t)\|\leq C(1+\log(1+F(t)))F(t),$$
so that
$$\frac{\dx}{\dx t}\log\left(1+\log(1+F(t))\right)\leq C.$$
Consequently, we find an estimate of the form
$$\|\phy(t)\|\leq F(t)\leq e^{ae^{b t}}.$$
Applying this inequality to $\phy^\eps$ with $G=\sH^2(\mathcal S)$, one gets a bound for the $\sH^2$ norm of $\phy^\eps$ on the interval $[0,T_{\max,\sH^2}^\eps)$, which is a contradiction with \eqref{liminfini}. The proof of \eqref{equa} is complete.

\bigskip
\noindent
{\em Step 4: global existence for $\eps$ small enough.}
Let $M_\eps=\|\phy^\eps\|_{\sL^2}=\|\phi_0^\eps\|_{\sL^2}$. By Lemma \ref{lemma-energy}, for $\eps\in(0,\eps_2)$, one has
\begin{align*}
\frac{1}{4}\|\partial_{x_1}\phy^\eps\|_{\sL^2}^2+\left(\frac{3}{8\eps^2}-C_0M_\eps^4\right)\|\partial_{x_2}(\mathsf{Id}-\Pi)\phy^\eps\|^2_{\sL^2}
&\leq \mathcal{E}_{\eps}(\phy^\eps)+C_0 M_\eps^2+C_0M_\eps^6\\
&\quad = \mathcal{E}_{\eps}(\phi_0^\eps)+C_0 M_\eps^2+C_0M_\eps^6.\\
\end{align*}
Hence, if $\eps M_\eps^2\leq (\frac{3}{16C_0})^{1/2}$, this inequality provides an $\sH^1$ bound for $\phy^\eps$ and, by \eqref{alternative}, we have $T_{\rm max}=+\infty$.
\end{proof}
\section{Reduction to the limit equation}\label{S:red}

This section is devoted to the proof of our two main theorems. As for the study of the Cauchy problem in Subsection \ref{cauchystrip}, we shall start with the case of $\sH^2$ initial data, which is simpler than the case of data in the energy space $\sH^1$ requiring an additional regularization argument.

\subsection{Proof of Theorem \ref{mainthmH2}}
\label{sectionmainthmH2}

Consider a sequence of Cauchy data $\phi_0^\eps\in \sH^2\cap \sH^1_0(\mathcal S)$ satisfying \eqref{ass3} and let $\theta^\eps(t)$ and $\phy^\eps(t)$ be respectively the solutions of \eqref{limit-equation} and \eqref{CNLS''}. Items {\em (i)} and  {\em (ii)} of Theorem \ref{mainthmH2} are direct consequences of Propositions \ref{wp-limit} and \ref{wp}. Notice that $\eps_1(M_0)$ is defined once for all by \eqref{eps1}.

Let us prove Item {\em (iii)}. To this aim, we first prove that Assumption \ref{assumption2} is satisfied, i.e. that the energy of $\phi_0^\eps$ is bounded from above. From \eqref{ass3} and \eqref{equivnorm}, one gets
\begin{equation}
\label{e1}
\left\|D_{x_1}^2\phi_0^\eps\right\|_{\sL^2}+\frac{1}{\eps^2}\left\|\left(D_{x_2}^2-\mu_1\right)\phi_0^\eps\right\|_{\sL^2}+\|\phi_0^\eps\|_{\sL^2}\leq C(1+M_2).
\end{equation}
Hence, by using \eqref{equivnorm2}, we have
\begin{equation}
\label{e2}
\|\mathcal P_{\eps,1}\phi_0^\eps\|_{\sL^2}\leq C,\qquad \|D_{x_{2}}\phi_{0}^\eps\|^2_{\sL^2}-\mu_{1}\|\phi_{0}^\eps\|^2_{\sL^2}\leq C\eps^2.
\end{equation}
Moreover, from the proof of Lemma \ref{lemma-energy}, we write
\begin{align}
\|\phi_0^\eps\|_{\sL^4}^4
&\leq C\|\Pi_1 \phi_0^\eps\|^4_{\sL^4}+C\|(\mathsf{Id}-\Pi_1) \phi_0^\eps\|^4_{\sL^4}\nonumber\\
&\leq C\|\phi_0^\eps\|_{\sL^2}^3\|\Pi_1\partial_{x_{1}}\phi_0^\eps\|_{\sL^2}+C\|\phi_0^\eps\|^2_{\sL^2} \|(\mathsf{Id}-\Pi_1)\partial_{x_{1}}\phi_0^\eps\|_{\sL^2} \|\partial_{x_{2}}(\mathsf{Id}-\Pi_1)\phi_0^\eps\|_{\sL^2}\nonumber\\
&\leq  C\|\phi_0^\eps\|_{\sL^2}^3\|\partial_{x_{1}}\phi_0^\eps\|_{\sL^2}+C\eps\|\phi_0^\eps\|^2_{\sL^2} \|\partial_{x_{1}}\phi_0^\eps\|_{\sL^2}\leq C\label{e3}
\end{align}
where we used \eqref{e2} and \eqref{spectral-gap}. Hence, \eqref{e1}, \eqref{e2} and \eqref{e3} yield $\mathcal E_\eps(\phi_0^\eps)\leq M_1$, for some $M_1>0$ independent of $\eps$: the sequence of Cauchy data satisfies Assumption \ref{assumption2}. Therefore, by conservation of mass and energy, for all $t\geq 0$, the sequence $\phy^\eps(t)$ also satisfies Assumption \ref{assumption2}. We can then apply Lemma \ref{cauchy-tensorial} to $\phy^\eps(t)$: for all $t\geq 0$ and for all $\eps\leq \eps_1(M_0)$, we have
\begin{equation}
\label{conft}
\|\phy^\eps(t)\|_{\sH^1(\mathcal S)}\leq C\quad\mbox{and}\quad \|(\mathsf{Id}-\Pi_1)\phy^\eps(t)\|_{\sL^2(\M,\sH^1(-1,1))}\leq C\,\eps.
\end{equation}

\bigskip
Let us now deal with the NLS equation \eqref{CNLS''} projected on $e_1(x_2)$: setting $u^\eps=\langle\phy^\eps(t),e_1\rangle_{\sL^2((-1,1))}$, we get
\begin{equation}
\label{projCNLS}
i\partial_{t}u^\eps=D_{x_{1}}^2u^\eps-\frac{\kappa^2(x_{1})}{4}u^\eps+\lambda\gamma |u^\eps|^2u^\eps+\langle R_{\eps}(\phy^\eps),e_1\rangle_{\sL^2((-1,1))}+\langle S_{\eps}(\phy^\eps),e_1\rangle_{\sL^2((-1,1))},
\end{equation}
with $u^\eps(0;\cdot)=\theta^\eps_0$ and where, for all $\phy\in \sH^2(\mathcal S)$, we have denoted
\begin{align}
R_\eps(\phy)&=m_\eps^{-1/2}D_{x_1}\left(m_\eps^{-1}D_{x_1}(m_\eps^{-1/2}\phy)\right)-D_{x_1}^2\phy-\frac{\kappa^2}{4}\left(m_\eps^{-2}-1\right)\phy\label{Reps}\\
S_\eps(\phy)&= \lambda m_\eps^{-1} |\phy|^2\phy-\lambda |\Pi_1\phy|^2\Pi_1\phy.\label{Seps}
\end{align}
Since $\theta^\eps=\theta^\eps(t,x_1)$ and $e_1=\Pi_1e_1$, we have
\begin{align*}
\|\phy^\eps(t)-\theta^\eps(t)e_1\|_{\sL^2(\mathcal S)}^2
&=\|\Pi_1(\phy^\eps(t)-\theta^\eps(t)e_1)\|_{\sL^2(\mathcal S)}^2+\|(\mathsf{Id}-\Pi_1)\phy^\eps(t)\|_{\sL^2(\mathcal S)}^2\\
&\leq \|u^\eps(t)-\theta^\eps(t)\|_{\sL^2(\M)}^2+C\eps^2
\end{align*}
by \eqref{conft}. Thus, to deduce \eqref{esti-erreur2}, it is enough to prove the following property: for all $T>0$, there exist $C_T>0$  and $\eps_T\in (0,\eps_1(M_0))$ such that, for all $\eps<\eps_T$, we have
\begin{equation}
\label{erru}
\forall t\in[0,T],\qquad \|u^\eps(t)-\theta^\eps(t)\|_{\sL^2(\M)}\leq C_T\,\eps.
\end{equation}
This fact will be a consequence of the following lemmas, that we prove further.
\begin{lemma}
\label{lemma2}
For all $\phy\in\sH^2(\mathcal S)$, we have the interpolation estimate
\begin{equation}
\label{interpolation}
\|\phy\|_{\sL^\infty}\lesssim \|\phy\|_{\sL^2}^{1/2}\|\phy\|_{\sH^2}^{1/2}.
\end{equation}
\end{lemma}
\begin{lemma}
\label{lemma1}
Let $\phy\in\sH^2(\mathcal S)$, then, for all $\eps\in (0,\eps_0)$,
$$\|R_\eps(\phy)\|_{\sL^2}\lesssim \eps\|\phy\|_{\sH^2}$$
and
$$\|S_\eps(\phy)\|_{\sL^2}\lesssim \|\phy\|_{\sL^2}\|\phy\|_{\sH^2}\|(\mathsf{Id}-\Pi_1)\phy\|_{\sL^2}+\eps \|\phy\|_{\sL^2}^2\|\phy\|_{\sH^2},$$
where $R_\eps$ and $S_\eps$ are defined by \eqref{Reps} and \eqref{Seps}.
\end{lemma}
\begin{lemma}
\label{lemma3}
Let $T>0$, let $\eps \in (0,\eps_0)$ and let $\phy^\eps\in C([0,T];\sH^2\cap \sH^1_0(\mathcal S))\cap C^1([0,T];\sL^2(\mathcal S))$ be solution of \eqref{CNLS''}. Assume moreover that we have an $\sL^\infty$ bound $\|\phy^\eps\|_{\sL^\infty([0,T]\times \mathcal S)}\leq M$, with $M$ independent of $\eps$. Then there exists $C_{M,T}>0$ such that, for all $t\in [0,T]$, we have
$$\left\|\left(\mathcal H_\eps-\frac{\mu_1}{\eps^2}\right)\phy^\eps(t)\right\|_{\sL^2}+\|\phy^\eps(t)\|_{\sL^2}\leq C_{M,T}\left(\left\|\left(\mathcal H_\eps-\frac{\mu_1}{\eps^2}\right)\phy^\eps(0)\right\|_{\sL^2}+\|\phy^\eps(0)\|_{\sL^2}\right).$$
\end{lemma}

\bigskip
\noindent
{\bf End of proof of Theorem \ref{mainthmH2}.}$\qquad$ In this proof, $C$ denotes a generic constant that only depends on the two upper bounds $M_0$ and $M_1$ in Assumption \ref{assumption2}. Consider the quantity
\begin{align*}
M=2\sup_{\eps\in (0,\eps_1(M_0))}\sup_{t\geq 0}\|\theta^\eps(t)\|_{\sL^\infty(\M)}&\leq C\sup_{\eps\in (0,\eps_1(M_0))}\sup_{t\geq 0}\|\theta^\eps(t)\|_{\sH^1(\M)}\\
&\leq C+C\sup_{\eps\in (0,\eps_1(M_0))}\|\theta_0^\eps\|_{\sH^1(\M)}^2\\
&\leq C+C\sup_{\eps\in (0,\eps_1(M_0))}\|\phi_0^\eps\|_{\sH^1(\mathcal S)}^2<+\infty,
\end{align*}
where we used the Sobolev embedding $\sH^1(\M)\hookrightarrow \sL^\infty(\M)$, the estimate \eqref{boundH1-theta}, Cauchy-Schwarz and the uniform bound \eqref{conft}. Next, for $\eps\in (0,\eps_1(M_0))$, by \eqref{conft}, \eqref{interpolation} and \eqref{e1} (which yields a uniform $\sH^2$ bound on $\phi_0^\eps$), we have
\begin{align}
\|\phi_0^\eps\|_{\sL^\infty}&\leq \|\phi_0^\eps-\theta_0^\eps \,e_1\|_{\sL^\infty}+\|\theta_0^\eps \,e_1\|_{\sL^\infty}=\|(\mathsf{Id}-\Pi_1)\phi_0^\eps\|_{\sL^\infty}+\|\theta_0^\eps\|_{\sL^\infty}\nonumber\\
&\leq C\|(\mathsf{Id}-\Pi_1)\phi_0^\eps\|_{\sL^2}^{1/2}\|(\mathsf{Id}-\Pi_1)\phi_0^\eps\|_{\sH^2}^{1/2}+\|\theta_0^\eps\|_{\sL^\infty}\nonumber\\
&\leq C\eps^{1/2}(1+M_2)^{1/2}+\frac{M}{2}.\label{t1}
\end{align}
Hence, for $\eps\leq M^2/(16C^2(1+M_2))$, one has $\|\phi_0^\eps\|_{\sL^\infty}\leq 3M/4$ and, by continuity of $\|\phy^\eps(t)\|_{\sL^\infty}$, we know that
\begin{equation}
\label{defTeps}
\end{equation}
belongs to $(0,+\infty]$. By a continuation argument, it is clear moreover that
\begin{equation}
\label{alterTeps}
\mbox{if}\quad T_\eps<+\infty \quad\mbox{then}\quad \|\phy^\eps(T_\eps)\|_{\sL^\infty}=M.
\end{equation}

Let us fix $T>0$. For all $t\leq \min(T,T_\eps)$, one has $\|\phy^\eps(t)\|_{\sL^\infty}\leq M$ so, from Lemma \ref{lemma3}, from \eqref{equivnorm} and from \eqref{ass3}, we deduce that
\begin{align*}
&\left\|D_{x_1}^2\phy^\eps(t)\right\|_{\sL^2}+\frac{1}{\eps^2}\left\|\left(D_{x_2}^2-\mu_1\right)\phy^\eps(t)\right\|_{\sL^2}+\|\phy^\eps(t)\|_{\sL^2}\leq \\
&\hspace*{5cm} \leq  C\left\|\left(\mathcal H_\eps-\frac{\mu_1}{\eps^2}\right)\phy^\eps(t)\right\|_{\sL^2}+C\|\phy^\eps(t)\|_{\sL^2}\\
&\hspace*{5cm}\leq C_{M,T}\left(\left\|\left(\mathcal H_\eps-\frac{\mu_1}{\eps^2}\right)\phi^\eps_0\right\|_{\sL^2}+\|\phi^\eps_0\|_{\sL^2}\right)\\
&\hspace*{5cm}\leq C_{M,T}\,(1+M_2).
\end{align*}
This yields the $\sH^2$ estimate
\begin{equation}
\|\phy^\eps(t)\|_{\sH^2}\leq \left\|D_{x_1}^2\phy^\eps(t)\right\|_{\sL^2}+\left\|\left(D_{x_2}^2-\mu_1\right)\phy^\eps(t)\right\|_{\sL^2}+(1+\mu_1)\|\phy^\eps(t)\|_{\sL^2}\leq CC_{M,T}(1+M_2).\label{e4}
\end{equation}
We can now apply Lemma \ref{lemma1}, together with \eqref{conft} and \eqref{e4} and, for all $t\leq \min(T,T_\eps)$, obtain
\begin{equation}
\label{e5}
\|R_\eps(\phy^\eps(t))\|_{\sL^2}+\|S_\eps(\phy^\eps(t))\|_{\sL^2}\leq \eps\,CC_{M,T}(1+M_2).
\end{equation}
Let us define $\widetilde u^\eps(t)=e^{itD_{x_1}^2}u^\eps(t)$ and $\widetilde \theta^\eps(t)=e^{itD_{x_1}^2}\theta^\eps(t)$ and write the equation satisfied by the difference $w^\eps(t)=\widetilde u^\eps(t)-\widetilde \theta^\eps(t)$:
\begin{align}
i\dr_{t}w^\eps=&e^{itD_{x_{1}}^2}\left(-\frac{\kappa^2(x_{1})}{4}\right)e^{-itD_{x_{1}}^2}w^\eps+\lambda\gamma \left(F(t;\widetilde u^\eps)-F(t;\widetilde\theta^\eps)\right)\nonumber\\
&+e^{itD_{x_{1}}^2}\left(\langle R_{\eps}(\phy^\eps),e_1\rangle_{\sL^2((-1,1))}+\langle S_{\eps}(\phy^\eps),e_1\rangle_{\sL^2((-1,1))}\right)\label{conj}
\end{align}
with $w^\eps(0)=0$. Hence, \eqref{lipF} together with the Sobolev embedding $\sH^1(\M)\hookrightarrow \sL^\infty$ and the bounds \eqref{conft}, \eqref{boundH1-theta}, \eqref{e5}, yield
$$\|\dr_{t}w^\eps(t)\|_{\sL^2}\leq C\|w^\eps(t)\|_{\sL^2}+\eps\,CC_{M,T}(1+M_2),\quad w^\eps(0)=0.$$
The Gronwall lemma gives then, for all $t\leq \min(T,T_\eps)$,
\begin{equation}
\label{e6}
\|u^\eps(t)-\theta^\eps(t)\|_{\sL^2}=\|w^\eps(t)\|_{\sL^2}\leq \eps\,CC_{M,T}(1+M_2)e^{CT}.
\end{equation}
We have proved the estimate \eqref{erru} for all $t\leq \min(T,T_\eps)$, and the proof of Theorem \ref{mainthmH2} will be complete if we show that there exists $\eps_T>0$ such that, for all $\eps\in(0,\eps_T)$, we have $T_\eps\geq T$.

Let us proceed by contradiction and assume that $T_\eps<T$. Apply as above the interpolation estimation \eqref{interpolation} at time $T_\eps$:
\begin{align*}
&\|\phy^\eps(T_\eps)\|_{\sL^\infty}\leq \|\phy^\eps(T_\eps)-\theta^\eps(T_\eps) \,e_1\|_{\sL^\infty}+\|\theta^\eps(T_\eps) \,e_1\|_{\sL^\infty}\\
&\quad\leq C\|\phy^\eps(T_\eps)-\theta^\eps(T_\eps) \,e_1\|_{\sL^2}^{1/2}\|\phy^\eps(T_\eps)-\theta^\eps(T_\eps) \,e_1\|_{\sH^2}^{1/2}+\|\theta^\eps(T_\eps)\|_{\sL^\infty}\\
&\quad \leq C\left(\|u^\eps(T_\eps)-\theta^\eps(T_\eps)\|_{\sL^2}+\|(\mathsf{Id}-\Pi_1)\phy^\eps(T_\eps)\|_{\sL^2}\right)^{1/2}\left(\|\phy^\eps(T_\eps)\|_{\sH^2}+\|\theta^\eps(T_\eps) \,e_1\|_{\sH^2}\right)^{1/2}\\
&\qquad +\|\theta^\eps(T_\eps)\|_{\sL^\infty}\\
&\quad \leq \eps ^{1/2}C(1+C_{M,T}(1+M_2)e^{CT})+\frac{M}{2}.
\end{align*}
where we used \eqref{e6}, \eqref{conft}, \eqref{e4}, \eqref{boundH2-theta} and the definition of $M$. Now we choose
$$\eps_T=\min\left(\eps_1(M_0),\left(C(1+C_{M,T}(1+M_2)e^{CT})\right)^{-2}\left(\frac{M}{4}\right)^2\right)$$
and obtain that, for all $\eps\in (0,\eps_T)$,
$$\|\phy^\eps(T_\eps)\|_{\sL^\infty}\leq \frac{3M}{4}<M.$$
Since $T_\eps<+\infty$, this contradicts \eqref{alterTeps}. The proof of Theorem \ref{mainthmH2} is complete.
\hspace*{0cm}\hfill\blacksquare

\begin{proofof}{Lemma \ref{lemma2}}
The interpolation estimate follows from the Sobolev embedding
\begin{equation}
\label{homog}
\forall u\in \sH^2(\R^2),\qquad \|u\|_{\sL^\infty}\lesssim \|u\|_{\sL^2}+\|D_{x_1}^2u\|_{\sL^2}+\|D_{x_2}^2u\|_{\sL^2}
\end{equation}
by a simple homogeneity argument. Indeed, for $\phy\in \sH^2(\R^2)$, non zero, inserting the function $u_\lambda$ defined by $u_\lambda(x)=\phy(\lambda x)$ with $\lambda = \|\phy\|_{\sL^2}^{1/2}\|\phy\|_{\sH^2}^{-1/2}$ in \eqref{homog} yields \eqref{interpolation}.
\end{proofof}

\begin{proofof}{Lemma \ref{lemma1}}
The estimate on $R_\eps(\phy)$ is immediate as soon as one notices that, for all $\eps<\eps_0$, we have $m_\eps\geq 1-\eps_0\|\kappa\|_{\sL^\infty}>0$ and that
\begin{equation}
\label{diffmeps}
\|m_\eps-1\|_{W^{2,\infty}(\mathcal S)}\leq \eps \|\kappa\|_{W^{2,\infty}(\M)}\leq C\eps.
\end{equation}
The estimate on $S_\eps$ follows from \eqref{diffmeps}, from
\begin{equation}
\label{calcSeps}|\phy|^2 \phy-|\Pi_1\phy|^2 \Pi_1\phy=(| \Pi_1\phy|^2+\phy\overline{ \Pi_1\phy})\,(\mathsf{Id}-\Pi_1)\phy+\phy^2\, \overline{(\mathsf{Id}-\Pi_1)\phy}.
\end{equation} and from the interpolation estimate \eqref{interpolation}.
\end{proofof}

\begin{proofof}{Lemma \ref{lemma3}}
Let us consider the time derivative of \eqref{CNLS''}: if $\chi^\eps=\dr_{t}\phy^\eps$, then
$$i\dr_{t}\chi^\eps=\left(\mathcal{H}_{\eps}-\eps^{-2}\mu_{1}\right)\chi^\eps+V_{\eps}\chi^\eps+2m_{\eps}^{-1}\lambda|\phy^\eps|^2\chi^\eps+\lambda m_{\eps}^{-1}(\phy^\eps)^2\overline{\chi^\eps}.$$
Take the $\sL^2(\mathcal S)$ scalar product with $\chi^\eps$ and then the imaginary part to get,
$$\frac{1}{2}\frac{\dx }{\dx t}\|\chi^\eps\|^2_{\sL^2}\leq C\int_{\mathcal S} |\phy^\eps|^2|\chi^\eps|^2\dx x_{1} \dx x_{2}\leq CM^2\|\chi^\eps\|^2_{\sL^2},$$
where we have used the assumption on the $\sL^\infty$ bound of $\phy^\eps$ on $[0,T]$. It remains to apply the Gronwall lemma:
\begin{align}
\|\dr_{t}\phy^\eps(t)\|_{\sL^2}
&\leq \|\dr_{t}\phy^\eps(0)\|_{\sL^2}\,e^{CM^2t}\nonumber\\
&\quad = \left\|\left(\mathcal H_\eps-\frac{\mu_1}{\eps^2}\right)\phy^\eps(0)+\left(V_{\eps}+\lambda m_{\eps}^{-1}|\phy^\eps(0)|^2\right)\phy^\eps(0)\right\|_{\sL^2}e^{CM^2t},\nonumber\\
&\leq \left(\left\|\left(\mathcal H_\eps-\frac{\mu_1}{\eps^2}\right)\phy^\eps(0)\right\|_{\sL^2}+C(1+M^2)\|\phy^\eps(0)\|_{\sL^2}\right)e^{CM^2t},\label{dtL2}
\end{align}
where we used the equation \eqref{CNLS''} to identify and bound $\dr_{t}\phy^\eps(0)$. We can now deduce an $\sH^2$ bound of $\phy^\eps$. Indeed, by \eqref{CNLS''}, we have 
\begin{align*}
\left\|\left(\mathcal H_\eps-\frac{\mu_1}{\eps^2}\right)\phy^\eps(t)\right\|_{\sL^2}
&=\left\|i\dr_{t}\phy^\eps-V_{\eps}\phy^\eps-\lambda m_{\eps}^{-1} |\phy^\eps|^2\phy^\eps\right\|_{\sL^2}\\
&\leq \|\dr_{t}\phy^\eps(t)\|^2_{\sL^2}+C(1+M^2)\|\phi^\eps(t)\|_{\sL^2}
\end{align*}
and the conclusion follows by using \eqref{dtL2} and $\|\phi^\eps(t)\|_{\sL^2}=\|\phi^\eps(0)\|_{\sL^2}$.
\end{proofof}

\subsection{Proof of Theorem \ref{mainthmH1}}
\label{sectionmainthmH1}

Consider a sequence of Cauchy data $\phi_0^\eps\in \sH^1_0(\mathcal S)$ satisfying Assumption \ref{assumption2}. Let $\theta^\eps(t)$ and $\phy^\eps(t)$ be respectively the global solutions of \eqref{limit-equation} and \eqref{CNLS''} (for $\eps\in (0,\eps_1(M_0))$. Items {\em (i)} and  {\em (ii)} of Theorem \ref{mainthmH1} stem from Propositions \ref{wp-limit} and \ref{wp}.

Let us prove Item {\em (iii)}. Since $\mathcal E_\eps(\phy^\eps(t))=\mathcal E_\eps(\phi_0^\eps)\leq M_1$ and by conservation of the $\sL^2$ norm, we deduce as in the proof of Theorem \ref{mainthmH2} that the estimates \eqref{conft} hold true, for all $t\geq 0$.

Let us regularize the initial data by setting
\begin{equation}
\label{regu}
\phi_0^{\eps,\eta}=\Pi_1 (1+\eta \eps D_{x_1}^2)^{-1/2}\phi_0^\eps,
\end{equation}
where $\eta>0$ is a parameter independent from $\eps$ that will be chosen later. It is clear that we have $\phi_0^{\eps,\eta}\in \sH^2\cap \sH^1_0(\mathcal S)$, with the following estimates:
$$\|\phi_0^{\eps,\eta}\|_{\sL^2}\leq \|\phi_0^\eps\|_{\sL^2},\quad \|D_{x_1}\phi_0^{\eps,\eta}\|_{\sL^2}\leq \|D_{x_1}\phi_0^\eps\|_{\sL^2},$$
$$\|D_{x_1}^2\phi_0^{\eps,\eta}\|_{\sL^2}\leq (\eta \eps)^{-1/2}\|\phi_0^\eps\|_{\sH^1},
$$
and that $(D_{x_2}^2-\mu_1)\phi_0^{\eps,\eta}=0$. In particular, we deduce from \eqref{equivnorm2}, \eqref{equivnorm} and \eqref{normL4} that
\begin{equation}
\label{borneseta}
\|\phi_0^{\eps,\eta}\|_{\sL^2}\leq M_0,\qquad \mathcal E_\eps(\phi_0^{\eps,\eta})\leq CM_1,\qquad \left\|(\mathcal H_\eps-\frac{\mu_1}{\eps^2})\phi_0^{\eps,\eta}\right\|_{\sL^2}\leq C(\eta \eps)^{-1/2}.
\end{equation}
Let $\phy^{\eps,\eta}(t)$ be the $\sH^2\cap \sH^1_0$ solution of \eqref{CNLS''} with the Cauchy data $\phi_0^{\eps,\eta}$ and let $\theta^{\eps,\eta}(t)$ be the solution of \eqref{limit-equation} with the Cauchy data $\langle\phi^{\eps,\eta}_0,e_1\rangle_{\sL^2((-1,1))}$. By Proposition \ref{wp} and since $\eps<\eps_1(M_0)$ and $\|\phi_0^{\eps,\eta}\|_{\sL^2}\leq M_0$, this solution $\phy^{\eps,\eta}(t)$ is defined for all $t\in \R_+$. Moreover, from $\mathcal E_\eps(\phy^{\eps,\eta}(t))=\mathcal E_\eps(\phi_0^{\eps,\eta})$ and the first two inequalities of \eqref{borneseta}, one gets again from Lemma \ref{cauchy-tensorial} that
\begin{equation}
\label{confteta}
\|\phy^{\eps,\eta}(t)\|_{\sH^1(\mathcal S)}\leq C\quad\mbox{and}\quad \|(\mathsf{Id}-\Pi_1)\phy^{\eps,\eta}(t)\|_{\sL^2(\M,\sH^1(-1,1))}\leq C\eps.
\end{equation}

\bigskip
\noindent
{\em Step 1: estimating $\phy^{\eps,\eta}(t)-\theta^{\eps,\eta}(t)e_{1}$.} Let us reproduce the series of estimates obtained in the proof of Theorem \ref{mainthmH2}. We have to take care to the fact that, here, the $\sH^2$ norm of $\phy^{\eps,\eta}(t)$ is not uniformly bounded but is of order $\eps^{-1/2}$. Defining $M$ by
$$M=2\sup_{\eps\in (0,\eps_1(M_0))}\sup_{t\geq 0}\|\theta^{\eps,\eta}(t)\|_{\sL^\infty(\M)}<+\infty,$$
and remarking that $\phi^{\eps,\eta}_0(x_1,x_2)=\theta^{\eps,\eta}_0(x_1)e_1(x_2)$, we get
$$
\|\phi_0^{\eps,\eta}\|_{\sL^\infty}=\|\theta_0^{\eps,\eta}\|_{\sL^\infty}\leq \frac{M}{2}.
$$
This enables to define $T_\eps>0$ similarly as above, by
$$T_\eps=\max\left\{t\geq 0\,:\,\mbox{for all }s\in [0,t],\,\|\phy^{\eps,\eta}(s)\|_{\sL^\infty}\leq M\right\}.$$
Next, Lemma \ref{lemma3} yields
$$\|\phy^{\eps,\eta}\|_{\sH^2}\leq CC_{M,T}(1+(\eta \eps)^{-1/2})$$
and \eqref{e5}, \eqref{e6} are now respectively replaced by
$$
\|R_\eps(\phy^{\eps,\eta})\|_{\sL^2}+\|S_\eps(\phy^{\eps,\eta})\|_{\sL^2}\leq \eps\,CC_{M,T}(1+(\eta \eps)^{-1/2})
$$
and, if we define $u^{\eps,\eta}=\langle\phy^{\eps,\eta}(t),e_1\rangle_{\sL^2((-1,1))}$, by
\begin{equation}
\label{esti1}
\|u^{\eps,\eta}(t)-\theta^{\eps,\eta}(t)\|_{\sL^2}\leq \eps\,CC_{M,T}(1+(\eta \eps)^{-1/2})e^{CT}\leq C_T\eps^{1/2}.
\end{equation}
This bound \eqref{esti1} holds true for all $t\leq \min (T,T_\eps)$. To show that $T_\eps\geq T$, we estimate again the $\sL^\infty$ norm of $\phy^{\eps,\eta}$ by interpolation, and obtain
\begin{align*}
&\|\phy^{\eps,\eta}(T_\eps)\|_{\sL^\infty}\leq \|\theta^{\eps,\eta}(T_\eps)\|_{\sL^\infty}+\\
&\quad + C\left(\|u^{\eps,\eta}(T_\eps)-\theta^{\eps,\eta}(T_\eps)\|_{\sL^2}+\|(\mathsf{Id}-\Pi_1)\phy^{\eps,\eta}(T_\eps)\|_{\sL^2}\right)^{1/2}\left(\|\phy^{\eps,\eta}(T_\eps)\|_{\sH^2}+\|\theta^{\eps,\eta}(T_\eps) \,e_1\|_{\sH^2}\right)^{1/2}\\
&\quad \leq \frac{M}{2}+C(\eta^{-1/4}\eps ^{1/4}+\eps^{1/2})(\eta^{-1/4}\eps^{-1/4}+1)\\
&\quad \leq \frac{M}{2}+C\eta^{-1/2}+C\eta^{-1/4}\eps^{1/4}+C\eps^{1/2}.
\end{align*}
Hence we first choose $\eta>0$ such that $C\eta^{-1/2}\leq \frac{M}{8}$. Then we choose $\eps_T>0$ such that $C\eta^{-1/4}\eps_T^{1/4}+C\eps_T^{1/2}\leq \frac{M}{8}$. Therefore, for all $\eps<\eps_T$, we have $\|\phy^{\eps,\eta}(T_\eps)\|_{\sL^\infty}\leq 3M/4$, which is sufficient to conclude as above that $T_\eps\geq T$.  Finally, by \eqref{esti1} and \eqref{confteta}, we have obtained that, for all $t\leq T$,
\begin{equation}
\label{esti2}
\|\phy^{\eps,\eta}(t)-\theta^{\eps,\eta}(t)e_1\|_{\sL^2}\leq \|u^{\eps,\eta}(t)-\theta^{\eps,\eta}(t)\|_{\sL^2}+\|(\mathsf{Id}-\Pi_1)\phy^{\eps,\eta}(t)\|_{\sL^2}\leq C\eps^{1/2}.
\end{equation}

\bigskip
\noindent
{\em Step 2: stability estimates.} Let us now estimate the two differences $\theta^\eps(t)-\theta^{\eps,\eta}(t)$ and $\phy^\eps(t)-\phy^{\eps,\eta}(t)$. By \eqref{regu}, we have
\begin{align}
\|\phi_0^\eps-\phi_0^{\eps,\eta}\|_{\sL^2}
&\leq \|(\mathsf{Id}-\Pi_1)\phi_0^\eps\|_{\sL^2}+\|\Pi_1 (1-(1+\eta \eps D_{x_1}^2)^{-1/2})\phi_0^\eps\|_{\sL^2}\nonumber\\
&\leq C\eps+\eta\eps\|D_{x_1}^2(1+(1+\eta \eps D_{x_1}^2)^{1/2})^{-1}(1+\eta \eps D_{x_1}^2)^{-1/2}\phi_0^\eps\|_{\sL^2}\nonumber\\
&\leq C\eps+\eta^{1/2}\eps^{1/2}\|D_{x_1}\phi_0^\eps\|_{\sL^2}\leq C\eps^{1/2}.\label{zinit}
\end{align}
The two functions $\theta^\eps(t)$ and $\theta^{\eps,\eta}(t)$ satisfy the same equation \eqref{limit-equation}, respectively with the initial data $\langle\phi^\eps_0,e_1\rangle_{\sL^2((-1,1))}$ and $\langle\phi^{\eps,\eta}_0,e_1\rangle_{\sL^2((-1,1))}$. Hence, since $\theta^\eps(t)$ and $\theta^{\eps,\eta}(t)$ are uniformly bounded in $\sL^\infty(\M)$ (because they are bounded in $\sH^1(\M)$), a standard stability estimate in $\sL^2$  yields, with the Gronwall lemma,
\begin{equation}
\label{esti3}
\|\theta^\eps(t)-\theta^{\eps,\eta}(t)\|_{\sL^2}\leq \|\phi_0^\eps-\phi_0^{\eps,\eta}\|_{\sL^2}\,e^{Ct}\leq C\eps^{1/2}e^{Ct}.
\end{equation}
To estimate the difference $z(t)=\phy^\eps(t)-\phy^{\eps,\eta}(t)$, we have to proceed in a different way, since $\phy^\eps(t)$ does not belong to $\sL^\infty(\mathcal S)$. Recall that  $\phy^\eps(t)$ and $\phy^{\eps,\eta}(t)$ satisfy the same equation \eqref{CNLS''}, with the initial data $\phi^\eps(0)=\phi^\eps_0$ and $\phy^{\eps,\eta}(0)=\phi^{\eps,\eta}_0$. Hence, the standard $\sL^2$ estimate on the difference $z^\eps(t)$ leads to
\begin{align}
\frac{\dx }{\dx t}\|z^\eps\|_{\sL^2}
&\leq |\lambda| \left\||\phy^\eps|^2|\phy^\eps|^2-|\phy^{\eps,\eta}|^2|\phy^{\eps,\eta}|^2\right\|_{\sL^2}\nonumber\\
&\leq \|\widetilde S_\eps(\phy^\eps)\|_{\sL^2}+\| \widetilde S_\eps(\phy^{\eps,\eta})\|_{\sL^2}+\||\Pi_1\phy^\eps|^2\Pi_1\phy^\eps-|\Pi_1\phy^{\eps,\eta}|^2\Pi_1\phy^{\eps,\eta}\|_{\sL^2}\label{esti5}
\end{align}
where $$\widetilde S_\eps(\phy)= \lambda |\phy|^2\phy-\lambda |\Pi_1\phy|^2\Pi_1\phy.$$ We now estimate $S_\eps(\phy^\eps)$ by coming back to \eqref{calcSeps}. By the H\"older and Minkowski inequalities and using that $\|\Pi\phy\|_{\sL^p}\lesssim \|\phy\|_{\sL^p}$, we get
\begin{align*}
\| \widetilde S_\eps(\phy^\eps)\|_{\sL^2}\leq C\|\phy^\eps\|_{\sL^{10}}^{5/2}\|(\mathsf{Id}-\Pi_1)\phy^\eps\|_{\sL^2}^{1/2}.
\end{align*}
Then, with a Sobolev embedding and \eqref{conft}, we deduce
$$\widetilde S_\eps(\phy^\eps)\|_{\sL^2}\leq C\|\phy^\eps\|_{\sH^1}^{5/2}\|(\mathsf{Id}-\Pi_1)\phy^\eps\|_{\sL^2}^{1/2}\leq C\eps^{1/2}.$$
Similarly, we also obtain $\|\widetilde S_\eps(\phy^{\eps,\eta})\|_{\sL^2}\leq C\eps^{1/2}$ thanks to \eqref{confteta}.
The last term in \eqref{esti5} is easy to estimate, since for all $\phy\in \sH^1_0(\mathcal S)$, one has, by Sobolev embedding in dimension one, $\|\Pi_1\phy\|_{\sL^\infty}\lesssim \|\phy\|_{\sH^1}$. Thus, by using again \eqref{conft} and \eqref{confteta}, we get,
\begin{align*}
\||\Pi_1\phy^\eps|^2\Pi_1\phy^\eps-|\Pi_1\phy^{\eps,\eta}|^2\Pi_1\phy^{\eps,\eta}\|_{\sL^2}
&\leq C\left(\|\Pi_1\phy^\eps\|_{\sL^\infty}+\|\Pi_1\phy^{\eps,\eta}\|_{\sL^\infty}\right)\|z^\eps\|_{\sL^2}\\
&\leq C\left(\|\phy^\eps\|_{\sH^1}+\|\phy^{\eps,\eta}\|_{\sH^1}\right)\|z^\eps\|_{L^2}\leq C\|z^\eps\|_{L^2}.
\end{align*}
Hence, \eqref{esti5} and \eqref{zinit} yield
$$
\frac{\dx }{\dx t}\|z^\eps(t)\|_{\sL^2}\leq C\eps^{1/2}+C\|z^\eps(t)\|_{\sL^2},\quad \|z^\eps(0)\|_{\sL^2}\leq C\eps^{1/2}
$$
and we conclude by the Gronwall lemma that 
\begin{equation}
\label{esti6}
\|\phy^\eps(t)-\phy^{\eps,\eta}(t)\|_{\sL^2}=\|z^\eps(t)\|_{\sL^2}\leq C\eps^{1/2}e^{Ct}.
\end{equation}
Finally, from \eqref{esti2}, \eqref{esti3} and \eqref{esti6}, one deduces the error estimate \eqref{esti-erreur}. The proof of Theorem \ref{mainthmH1} is complete.\hspace*{0cm}\hfill\blacksquare

\bigskip
\subsection*{Acknowledgements}
We wish to thank Christof Sparber for helpful discussions. F.M. acknowledges support by the ANR-FWF Project Lodiquas (ANR-11-IS01-0003) and by the ANR Project Moonrise (ANR-14-CE23-0007-01).


\end{document}